\documentclass[11pt]{article}
\usepackage[margin=1in]{geometry}
\usepackage{paralist,multirow,amsmath,amsthm,algorithmic,algorithm,url,graphics}
\newtheorem{definition}{Definition}
\newtheorem{theorem}{Theorem}

\title{Implementing Support for Pointers to Private Data in a
  General-Purpose Secure Multi-Party Compiler}
\date{}
\author{Yihua Zhang\\
Department of Computer Science and Engineering\\
University of Notre Dame\\
yzhang16@nd.edu
\and Marina Blanton\\
Department of Computer Science and Engineering\\
State University of New York at Buffalo\\
mblanton@buffalo.edu
\and Ghada Almashaqbeh\footnote{Work done while at the University of Notre Dame.}\\
Department of Computer Science\\
Columbia University\\
ghada@cs.columbia.edu}
  
\begin{document}
\maketitle
\begin{abstract}
Recent compilers allow a general-purpose program (written in a
conventional programming language) that handles private data to be
translated into secure distributed implementation of the corresponding
functionality. The resulting program is then guaranteed to provably
protect private data using secure multi-party computation techniques.
The goals of such compilers are generality, usability, and efficiency,
but the complete set of features of a modern programming language has
not been supported to date by the existing compilers. In particular,
recent compilers PICCO and the two-party ANSI C compiler strive to
translate any C program into its secure multi-party implementation,
but currently lack support for pointers and dynamic memory allocation,
which are important components of many C programs. In this work, we
mitigate the limitation and add support for pointers to private data
and consequently dynamic memory allocation to the PICCO compiler,
enabling it to handle a more diverse set of programs over private
data. Because doing so opens up a new design space, we investigate the
use of pointers to private data (with known as well as private
locations stored in them) in programs and report our findings. Besides
dynamic memory allocation, we examine other important topics
associated with common pointer use such as reference by
pointer/address, casting, and building various data structures in the
context of secure multi-party computation. This results in enabling
the compiler to automatically translate a user program that uses
pointers to private data into its distributed implementation that
provably protects private data throughout the computation. We
empirically evaluate the constructions and report on performance of
representative programs.
\end{abstract}

\section{Introduction} 

Recent advances in secure multi-party computation make it feasible to
securely compute with private data belonging to different organizations even
for complex functionalities. Furthermore, together with ubiquitous
proliferation of cloud computing services, these techniques give rise to
secure computation outsourcing. For these reasons, the research community
has recently developed a number of compilers for transforming a
general-purpose program into the corresponding secure distributed
implementation (see, e.g., \cite{hol12,zha13}). These tools aim at
generality and are designed to translate a program written in a conventional
programming language into an equivalent program that uses secure computation
techniques to protect private data. They also aid usability and make it
easier for a programmer without extensive knowledge of secure computation
techniques to produce a protocol that can be securely executed in a
distributed environment.

It has been long known that any computable function can be securely
evaluated by multiple participants if it is represented as an arithmetic or
Boolean circuit. This representation, however, is not always obvious or
known or may otherwise significantly increase the program size. Existing
compilers remove the need for the programmer to perform this translation
manually and assemble secure implementations from efficient building blocks
for elementary operations. Thus, efficiency of the resulting secure
computation is also one of the goals that compilers target. Furthermore, the
ability to support both private (i.e., protected) and public (i.e., not
protected) data or variables in a single program adds a level of complexity
to the implementation because of the need to support interaction between
public and private variables and secure data flow enforcement.

While the design goal of several compilers was to support any feature
of a general-purpose programming language (such as C in \cite{hol12} and
\cite{zha13}), all such compilers we are aware of have limitations. In
particular, the original version of the PICCO compiler \cite{zha13} provided
no direct support for C pointers (i.e., pointers were supported only in the
form of static arrays) and, as a result, no support for dynamic memory
allocation other than static arrays. Similarly, the original 
version of the two-party compiler for ANSI C \cite{hol12} supported pointers
only in the form of statically allocated arrays restricted to a constant size
and had additional limitations (such as support for floating point
arithmetic was not available in the open source CBMC that the compiler
builds upon). Thus, support for C-like pointers -- or, in other programming
languages, support for the features that pointers enable such as dynamic
memory allocation, reference by pointer or address, and building data
structures -- is the most crucial part of a general-purpose program that is
currently unavailable in existing compilers. Adding this support is thus the
focus of our work.

In this paper, we extend the PICCO compiler \cite{zha13} with pointer
support. PICCO\footnote{Available from GitHub at
https://github.com/picco-team/picco.} is a source-to-source translator that
takes as an input a program written in the C programming language with
variables to be protected marked as private and produces a C program that
implements the computation using secure multi-party computation techniques
based on linear secret sharing. We view PICCO as an attractive compiler
choice because of the flexibility of the setting it uses. In particular, the
setting assumes three groups of participants: (i) input parties who hold
private inputs into the computation, (ii) computational parties who perform
secure computation on secret-shared data, and (iii) output parties who are
entitled to learning the result of the computation. The composition of these
three groups can be arbitrary (in particular, including the same,
overlapping, or non-overlapping groups), which makes the setting suitable
for secure multi-party computation (SMC), delegation of the computation by
multiple data owners to a subset of them or other suitable entities or
secure computation outsourcing by one or more parties. This flexibility
follows from the use of secret sharing techniques and may or may not be
present in tools that build on alternative secure computation techniques
(such as, e.g., garbled circuit evaluation).

With linear secret sharing, before secure computation can commence, each
input party splits her private inputs into $n > 2$ secret shares, where $n$
is the number of computational parties, and communicates each share to a
respective computational party. The computational parties then proceed with
evaluating the function on secret-shared data and communicate their shares
of the result to the output parties who reconstruct the output using their
shares. Any linear combination of secret-shared integers is performed
locally, but multiplication of secret-shared integers constitutes the
elementary interactive operation. Performance is then measured in the total
number of interactive operations as well as the number of sequential
interactions or rounds, and recent solutions based on secret sharing aim at
minimizing overhead using both metrics.

When PICCO is used to perform source-to-source translation, the input
program is a conventional C program where each variable is marked to be
either private or public. All computation with private variables is
transformed into secure arithmetic on shared data, while operations with
public variables that do not interact with private data are left unchanged.
In addition to specifying private/public qualifies for each variable, for
performance reasons PICCO also allows the programmer to mark the places
where computation can proceed concurrently (i.e., to decrease the number of
computation rounds), which also extends the conventional C syntax. 

Adding pointer to support to a program that manipulates private data not
only extends the compiler to handle the full range of C programs (that do
not violate secrecy of private data), but also permits important features of
programming languages, treatment of which, to the best of our knowledge, has
not been done before. As part of this work, we thus explore how pointers to
private data (including pointers with private locations) can be implemented
and discuss our design decisions. Having added support of
pointers to private data, we further study common uses of pointers in
programs and the impact our implementation has on those language features.
For example, we evaluate passing arguments by references, dynamic memory
allocation, and pointer casting. Based on our analysis as well as empirical
evaluation, several of these features introduce only marginal costs. Also, one
of the important topics studied in this is work is the use of pointer-based
data structures written for private data. Our results indicate that the use
of pointers (to private data) is very attractive and maintains high
efficiency for several popular data structures. In some other cases, in
particular when working with sorted data, privately manipulating pointers
increases complexity of data structure operations and it might be desirable
to pursue alternative implementations. 

We would like to emphasize that it is not the goal of this paper to try to
develop most efficient implementations for different data structures.
Instead, the goal is to determine how pointers to private data can be
supported at a low possible cost and to what performance of typical
programs that might lead. We note that, depending on the program structure,
asymptotic complexity of a translated program might be higher than that of
the original. For example, consider an if-then-else statement with a private
condition (e.g., conditional statements used in traversing a binary tree).
When data privacy is not required, only one of the two branches will be
executed, but with any compiler that produces a secure implementation both
branches will have to be evaluated to hide the result of the private
condition. Then with a sequence of $n$ nested if-then-else statements, in
the worst case the secure program might have to execute $O(2^n)$
instructions where the original program would execute only $O(n)$. This
means that the general translation approach can lead to an exponential
increase in the runtime for programs of practical relevance. As part of this
work we show that data structures that utilize pointers to private data
cover the entire spectrum of possibilities: in one extreme, they result in
no asymptotic increase over conventional non-secure counterparts, and in
another extreme, the increase is exponential. This provides insights on when
natural pointer use is very attractive and when other, alternative
implementations might be desired.

While for many data structures alternative, non-pointer-based
implementations may be possible, we note that our extension of PICCO with
pointers enables support for an important aspect of modern programs
otherwise not available in any secure multi-party compiler we are aware of,
which is dynamic memory allocation. Dynamic memory allocation is essential
for a general-purpose programming language, but has not been systematically
studied in the context of secure multi-party computation (e.g., even
publications and compilers that run secure computation on data sets of a
large size assume that the size of the data is fixed and known at
compilation time). As an example, consider an application in which data
items arrive over time. A pointer-based linked list will allow for a natural
and graceful mechanism of memory management, that unlike statically
allocated arrays does not require complex provisions for allocating a larger
buffer when the current one becomes full, moving the data, or merging
previously allocated buffers.

The rest of this paper is organized as follows: We first give a brief
overview of related work in Section~\ref{sec:rel-work}. After presenting
background information in Section~\ref{sec:background}, we proceed with
presenting our solution for supporting pointers to private data in
Section~\ref{sec3}. In Section~\ref{sec:pointer-uses}, we discuss common
uses of pointers in programming, such as passing arguments by reference,
dynamic memory allocation, array manipulation, and pointer casting and their
underlying implementation in our framework. Section~\ref{sec:analysis} next
summarizes operations with pointers to private data and formally shows
security of the design. Section~\ref{sec:data-structures} analyzes various
data structures built using pointers to private data. Lastly,
Section~\ref{sec:perf} presents the results of performance evaluation of
representative programs that utilize pointers (to private data) and
Section~\ref{sec:conclusions} concludes this work.

\section{Related Work}
\label{sec:rel-work}

In this section we review the most closely related work on SMC compilers and
secure/oblivious data structures. Regarding the compilers,
Fairplay~\cite{Malkhi04} was a pioneer work that enables compilation of
secure two-party protocols based on garbled circuits. Its extension to
multiple parties, FairplayMP~\cite{Ben08}, implements secure computation
using Boolean circuits and secret sharing techniques. TASTY~\cite{Henecka10}
is another two-party SMC compiler that combines garbled circuit techniques
with those based on homomorphic encryption. Sharemind \cite{Bogdanov08} and
VIFF~\cite{Damgaard09} are multi-party compilers based on custom additive
3-party secret sharing and standard threshold linear secret sharing,
respectively. All of the above compilers use custom domain-specific
languages to represent user programs. The two-party compiler for ANSI
C~\cite{hol12} and PCF~\cite{Kreuter13} both use two-party garbled circuit
techniques, where the former's goal is to support general purpose C
programs, while the latter uses a new circuit format and employs
optimizations to reduce the compilation time and storage. Lastly,
TinyGarble~\cite{son15} uses hardware synthesis to optimize garbled circuits
for two-party computation. All of these compilers require linear in the size
of memory work to access memory at a private location. SCVM~\cite{Liu14}, on
the other hand, is an automated compiler that utilizes oblivious RAM (ORAM)
and targets two-party computation. ObliVM~\cite{Liu15} is another ORAM-based
secure two-party computation compiler that transforms programs written in
high level abstractions to optimized garbled circuit implementations.
Finally, a recent compiler Frigate~\cite{moo16} was designed to guarantee
correctness of programs compiled into circuits for secure two-party
computation.

To support data structures in the SMC framework, several solutions
\cite{Toft11,Keller14,Mitchell14,Wang14} have been proposed. The main
motivation of this line of work is the need to store and manipulate
private data in an efficient and flexible manner. Toft~\cite{Toft11}
proposed a private priority queue that has a deterministic access
pattern as opposed to randomized ones in ORAM-based data structures. 
On the other hand, Keller and Scholl \cite{Keller14} introduced
implementations of arrays, dictionaries, and priority queues based on
various flavors of ORAM implementations. Mitchell and Zimmerman
\cite{Mitchell14} also provide implementations of stacks, queues, and
priority queues based on oblivious data compaction and an offline
variant of ORAM. Wang et al. \cite{Wang14} proposed implementations of
maps, sets, priority queues, stacks, and deques based on ORAM
techniques modified for specific data access patterns. Different from
all of these publications, our work includes extending the PICCO
compiler to support dynamic data structures in a generic way as found
in general purpose programming languages. That is, the programmer has
the basic tools and primitives that enable her to build any desired
data structure.

In our implementation, a pointer to private data may store one or more
locations where the data might reside, which in the worst case is linear in
the program's memory size. ORAM-based techniques, on the other hand,
guarantee that when an item is accessed at a private location, the number of
accessed memory locations is polylogarithmic in the total memory size. Thus,
our general solution may or may not be faster than using ORAM, depending on
both the program and the data size. As we discuss later in this work,
employing ORAM techniques can be beneficial for certain data structures (and
sufficiently large data sets). Building custom data structures, however, is
beyond the scope of this work. 

One of the applications that the compiler can naturally be used for
once support for pointers to private data is in place is evaluation of
a context-free grammar on private data (implemented as a shift-reduce
parser using a stack). The grammar can be either public or private,
and in the latter case execution will correspond to evaluation of
private expressions/programs on private data. Techniques for
evaluation of private programs (on private data) are a separate area
of research, discussion of which is beyond the scope of this work, but
the reader may refer to recent results in this areas such as those in
\cite{kis16,son16}.

\section{Background Information}
\label{sec:background}

PICCO uses Shamir secret sharing \cite{sha79} for implementing secure
arithmetic and other operations on private data. It is an $(n,t)$-threshold
linear secret sharing scheme, in which a private value is represented using
$n > 2$ secret shares, one held by each computational party. Then any $t+1$
or more shares can be used to reconstruct the private value, while $t$ or
fewer parties cannot learn any information about the shared value (which is
perfectly protected in the information-theoretic sense). In a linear secret
sharing scheme, a linear combination of secret-shared values can be
performed by each computational party locally, without any interaction,
while multiplication of secret-shared values requires communication between
all of them. With Shamir secret sharing, computation takes place over a
field of a desired size (larger than any value that needs to be
represented). A secret $s$ is represented using a random polynomial of
degree $t$ with the free coefficient set to $s$, and each share corresponds
to the evaluation of the polynomial on a distinct non-zero point. Given
$t+1$ or more shares, the secret can be reconstructed using Lagrange
interpolation. Then addition or subtraction of secret-shared values, or
multiplication of a secret-shared value by a known integer can be performed
by each party locally using its shares. Multiplication involves multiplying
two shares, which raises the corresponding polynomial degree to $2t$, and
resharing and interpolating the result to bring the degree of the
corresponding polynomial from $2t$ to $t$. This imposes the requirement that
$t < n/2$. With the way multiplication is performed, it is also possible to
evaluate any multi-variate polynomial of degree 2 over secret-shared
integers with a single interaction. That is, we first evaluate the
polynomial and re-share the overall result instead of doing so for the
intermediate products. This serves as a powerful optimization tool when, for
instance, the dot product of two secret shared vectors/arrays needs to be
computed, which results only in a single interaction.

While the regular field operations achieve perfect secrecy, implementation
of some of the basic operations used in PICCO (such as comparisons,
division, etc.) is statistically secure, which requires the bitlength of
the field elements to be increased by the statistical security parameter.
This slightly increases the cost of field operations, as well as the amount of
communication associated with transmitting field elements. The optimal size of
field elements is automatically determined by PICCO for each program it
compiles.

Performance of these techniques is measured in terms of the total number of
elementary interactive operations (field multiplications or reconstructions
of a value from its shares) as well as the number of sequential interactions or
rounds. For that reason, PICCO supports a number of optimizations to reduce
the round complexity of programs that it outputs through concurrent or
batch execution. 

Support for pointers to private data (and the corresponding functionalities
such as dynamic memory management) was the only missing functionality in
PICCO. Thus, we modify the compiler to enable it to compile user programs
that contain pointers to private data, which was not previously available.
Our changes to the compiler affect only pointers to private data and the
introduction of two built-in functions for memory allocation and
deallocation (called \texttt{pmalloc} and \texttt{pfree}) associated with
pointers to private data.

\section{Adding Pointer Support}
\label{sec3}

Recall that in C a pointer is a variable of a special type that stores a
location in memory at which data of a particular type can be located.
Because a pointer stores an address, it can be treated very generally with
the possibility of directly manipulating pointers, changing the addresses
they store, dereferencing a pointer to access the data to which it points,
casting a pointer of a particular type to a pointer of another type, and
using pointers to functions.

\subsection{Working Toward a Solution}

When working with pointers in the presence of private data, besides
traditional C pointers to public variables, we can distinguish between
pointers to private data that (i) point to a single known location where the
private data is stored and (ii) point to a memory pool or a number of
locations where private data is stored and the location of the private data
is not known. In determining how this can be implemented in a C-like
programming language, we considered pre-allocating memory pools for pointers
with private locations. Such memory pools would be required for each data
type to ensure that we can store and extract private data correctly. This
approach, however, has severe disadvantages, which are: 
\begin{enumerate}
  \item Using memory pools unnecessarily increases the program's memory
    footprint, where one pool will be needed for each used data type
    including complex types defined via the struct construct. Furthermore,
    it is not clear to what size each pool should be set to optimize
    performance. 
  \item It would also often incur unnecessarily large computation
    costs due to the need to touch all locations within a memory pool per
    single access (or touch several locations when the pool is implemented
    using more complex ORAM techniques). As will be evident later in the
    paper, there are large classes of programs, applications, and data
    structures, where a pointer to private data always corresponds to a
    single location, which removes the need to use secure multi-party
    computation techniques for pointer manipulation. Allowing a pointer to
    store a single known location drastically improves program performance
    compared to using pointer pools.
  \item Memory pools would also not work in the presence of pointer
    casting. 
\end{enumerate}

Then if we do not want a pointer to initially point to a pre-allocated memory
pool, would the decision to properly declare a pointer as pointing to a
single (known) location or a set of locations be left to the programmer?
This is going to introduce an additional burden for a programmer who would
need to know at a variable declaration time whether the variable of a
pointer type will require protecting its value. This happens if the pointer
is used inside a conditional statement with private condition, which then 
requires protecting the location assigned to the pointer to protect the
result of the condition evaluation. 

To ease programming burden and at the same time avoid consuming
unnecessary (memory and computation) resources, our solution is to use
the same programming interface for all pointers that are to point to
private data. When the pointer is being declared or initialized, it
has one known location associated with it (if the pointer is not
initialized, that location is set to the default value corresponding
to uninitialized pointers). Throughout the computation, the pointer,
however, may be pointing to multiple locations, one of which is
its true location. This happens when the pointer's value is modified
inside conditional statements with private conditions as illustrated
next. Suppose we declare variables \texttt{a} and \texttt{b} to be
private integers followed by the code below:

{\small \begin{verbatim}
1. private int *p;
2. p = &a;
3. if (priv-cond) then p = &b;
\end{verbatim}}
\noindent We see that variable \texttt{p} was declared as a pointer to
a private integer, but the type of the pointer with respect to whether
the location itself is private is implicit. After executing 
lines 1--2, \texttt{p} has a single known location,
but after executing line 3, \texttt{p} is associated with a list of
two locations (the address of \texttt{a} and the address of
\texttt{b}) and the value of the true location is protected. That
is, a pointer always starts with a single publicly known location and the
location to which it is pointing may become private, but the user does not
declare the pointer itself as public or private. In the
rest of this work, we use the term ``public location'' in reference to
a pointer to private data to mean that the pointer has a single known
location (either initialized or uninitialized) and we use the term
``private location'' to mean that the pointer has a list of public
locations, but which location is in use remains private.

When we consider interaction of public and private values in connection to
the use of pointers, a number of questions arise, which we address next.
\begin{enumerate}
\item \emph{Can a pointer that was declared to point at private data be
    assigned address of public data?} Note that without the use of pointers,
  the equivalent actions are generally allowed. That is, a variable declared
  to hold private data can be assigned a known value, which is consequently
  converted into protected form. The same does not hold for pointers and we
  disallow assigning locations of public variables to pointers which were
  declared to point to private data. To see why, suppose that a user program
  contains the code below where \texttt{a} was declared to be a private
  integer, while \texttt{b} is a public integer:
  
{\small \begin{verbatim}
1. p = &a;
2. if (priv-cond) then p = &b;
3. *p += 1;
\end{verbatim}}
After executing lines 1--2, \texttt{p} stores two addresses and the
true location of where it is pointing out of these two addresses is
protected. On line 3, however, the pointer is dereferenced and the
result of private condition \texttt{priv-cond} evaluation is revealed
by examining the value of \texttt{b} before and after line 3. Thus, to
eliminate information leakage, pointers to private data can be
assigned only locations that store private values.

\item \emph{Can a pointer declared to point to public data be modified
    inside conditional statements with private conditions and as a result
  become pointing to multiple locations?} The answer to this question is No.
  If a pointer to public data is updated in the body of a conditional
  statement with private condition, it must be treated as a pointer to
  private data (otherwise, using its dereferenced value reveals unauthorized
  information). Allowing such uses and performing the conversion implicitly
  by the compiler will be confusing to the programmer (who no longer can use
  the pointer to store addresses of public data). For that reason, we
  disallow updates to pointers to public data within the body of conditional
  statements with private conditions.
\end{enumerate}

\subsection{Pointer Implementation}
\label{sec:pimpl} 

We next proceed with describing how pointers to private data are
implemented to realize the ideas outlined above. We note that all
program transformations that we describe preserve semantics of the
original program and, given that a program can be compiled into the
corresponding secure implementation, the transformed program will
always produce the same output as the original program. There are some
restrictions that user programs must meet in order to be compiled into
secure implementations with no information leakage. Such restrictions
include the two cases at the interaction of public and private data
described above and additional restrictions inherited from PICCO
(e.g., the fact that the body of a conditional statement with a
private condition cannot have public side effects, a loop termination
condition should be public or made public, etc.). This is to ensure
that no information leakage in the compiled program can take place,
and the programs that do not meet the requirements are aborted at the
compilation time. Once these constraints are met, our extension of
PICCO will allow any user program to be compiled into its secure
counterpart.

\medskip \noindent \textbf{Pointer representation.}
As we incorporate support for pointers, we first note that pointers to
public data will not need to be modified and their implementation remains
the same as in the C programming language. The most significant change in
implementing pointers to private data comes from the need to maintain
multiple locations. For that reason, the data structure that we maintain
for pointers to private data consists of (i) an integer field that stores
the number $\alpha$ ($\ge 1$) of locations associated with the pointer; (ii)
a list of $\alpha$ addresses where the data is stored; and (iii) a list of
$\alpha$ private (i.e., secret-shared) tags, one of which is set to 1 (true
data location) and all others are set to 0. For the important special case of
$\alpha=1$, the pointer has known (public) location and the tags are not
used. 

We formalize the above pointer representation using the following invariant,
which is maintained throughout various pointer operations: \emph{among all
locations stored with a pointer to a private object, there is exactly one
true location of the object and the tag corresponding to that location is
set to 1, while the tags corresponding to all other locations are set to 0.}
This invariant is true of all well-formed programs and may be violated only
in the case of dangling pointers as detailed later.

Because we would like to employ a uniform data structure for pointers to
private data of any data type such as integer, floating point values, etc.
and even pointers to a pointer, the data structure we maintain needs to
include two additional fields: (iv) an integer flag that determines the type
of data associated with the pointer (i.e., integer = 1, float = 2, struct =
3, etc.) and (v) an integer field that indicates the indirection level of
the pointer. For instance, if a pointer refers to a private value of a
non-pointer type, its indirection level is set to 1; and if it refers to a
pointer whose indirection level is $k$ (for $k \ge 1$), its level will be
set to $k+1$. A pointer to a struct also has indirection level 1 regardless
of the types of the struct's fields (which can be pointers themselves). 

\medskip \noindent \textbf{Pointer updates.}
Initially, at the pointer declaration time, the number of locations $\alpha$
associated with the pointer is set to 1 and the address is set to to a
special constant used for uninitialized pointers. Then every time the
pointer is modified (including simultaneously with pointer declaration), its
data structure is updated. When the pointer is assigned a new location using
a public constant, a variable's address, or a memory allocation mechanism
(e.g., as in \texttt{p = 0}, \texttt{p = \&a}, or \texttt{p =
malloc(size)}), $\alpha$ in the pointer's data structure is set to 1 and the
associated address is stored in the pointer's address list. When a pointer
is updated using another pointer (as in \texttt{p = p1}), the latter's
data structure is copied and stored with the former. 

Such simple manipulations are used only when the assignment does not take
place inside the body of a conditional statement with a private condition.
Pointer assignments inside conditional statements with a private condition
present the most interesting case when the list of pointer locations gets
modified. Updating values modified in the body of a conditional statement
with a private condition already requires special handling in PICCO, and all
we need is to support a specific procedure when a variable of pointer type
is being modified. We need to distinguish between if-then and if-then-else
statements, which we consequently discuss. 

Consider the following code with an if-then statement:

{\small 
\begin{verbatim}
1. p = p1;
2. if (priv-cond) then p = p2;
\end{verbatim}}
\noindent where \texttt{p}, \texttt{p1}, and \texttt{p2} are pointers
(to private data) of the same type. This is the most general case,
where on line 2 both \texttt{p} and \texttt{p2} can have any number of
locations associated with each of them (recall that all other
assignment types use a single location). When this code is written for
ordinary (private) variables of the same type $a$, $a_1$, and $a_2$, a
generic way to implement this update in PICCO and similar compilers is
to first set $[a] = [a_1]$ and then compute $$[a] = [c] \cdot [a_2] +
(1 - [c]) \cdot [a] = [c] \cdot ([a_2] - [a]) + [a],$$ where $c$ is a
bit equal to the result of evaluating \texttt{priv-cond}. We use
notation $[x]$ to indicate that the value of $x$ is protected via
secret sharing and computation takes place on its shares. In the case
of pointers, such a simple update does not work because this procedure
would turn addresses into secret shared values preventing the pointer
from being dereferenced (without touching all possible memory
locations). Thus, after executing the assignment \texttt{p = p1}, we
combine the (public) locations of \texttt{p} and \texttt{p2} and set
the tags in \texttt{p} based on the current tags of \texttt{p} and
\texttt{p2} and the result $c$ of evaluating \texttt{priv-cond}. Let
pointer \texttt{p} after executing the first assignment contain
$\alpha_1$ locations stored as $L_1 = \{\ell_1, {\ldots},
\ell_{\alpha_1}\}$ with corresponding tags $T_1 = \{[t_1], {\ldots},
[t_{\alpha_1}]\}$ (i.e., this information was copied from
\texttt{p1}). Similarly, let pointer ${p_2}$ store $\alpha_2$, $L_2 =
\{\ell'_1, {\ldots}, \ell'_{\alpha_2}\}$, and $T_2 = \{[t'_1],
{\ldots}, [t'_{\alpha_2}]\}$. Note that the ordering of addresses in
each $L$ is arbitrary, but the tag $t_i$ in $T$ must correspond to the
address $\ell_i$ at the same position $i$ in $L$. Then as a result of
the conditional assignment, we compute \texttt{p}'s new content as
given in Algorithm~\ref{alg:cond-as}.
\begin{algorithm}[h]
\caption{$\langle \alpha_3, L_3, T_3 \rangle \leftarrow {\sf
    CondAssign}(\langle \alpha_1, L_1, T_1 \rangle, \langle \alpha_2$, $L_2,
  T_2 \rangle, [c])$} \label{alg:cond-as}
\begin{algorithmic}[1]
    \STATE $L_3 = L_1 \cup L_2$;
    \STATE $\alpha_3 = |L_3|$; 
    \FOR{every $\ell''_i \in L_3$}
    \STATE $pos_1 = L_1.{\sf find}(\ell''_i)$;
    \STATE $pos_2 = L_2.{\sf find}(\ell''_i)$;
    \IF{$(pos_1 \not= \perp$ and $pos_2 \not= \perp)$}
    \STATE $[t''_i] = [c] \cdot [t'_{pos_2}] + (1 - [c]) \cdot [t_{pos_1}]$;
    \ELSIF{($pos_2 = \perp$)}
    \STATE $[t''_i] = (1 - [c]) \cdot [t_{pos_1}]$;
    \ELSE
    \STATE $[t''_i] = [c] \cdot [t'_{pos_2}]$;
    \ENDIF
    \ENDFOR
    \STATE set $T_3 = \{[t''_1], [t''_2], {\ldots}, [t''_{\alpha_3}]\}$;
    \STATE return $\langle \alpha_3, L_3, T_3 \rangle$;
\end{algorithmic}
\end{algorithm}

In the algorithm, $L_3$ is composed of all locations
appearing in $L_1$ or $L_2$ (repeated locations are stored only once).
We use notation $L.{\sf find}$ to retrieve the position of the element
of $L$ provided as the argument or special symbol $\perp$ is the
element is not found. The tags in the output $T_3$ are set based on
three different cases: (i) a location in $L_3$ is found in both $L_1$
and $L_2$; (ii) it is found in $L_1$, but not in $L_2$; and (iii) it
is found in $L_2$, but not $L_1$. Because only tags in $T_1$ and $T_2$
and $c$ are private, only lines 7, 9, and 11 correspond to private
computation.
 
If the conditional statement is of the form if-then-else, but \texttt{p} is
not updated in the body of the else clause, then the computation in
Algorithm~\ref{alg:cond-as} is applied unchanged. If the pointer is instead
updated only in the body of the else clause, then the computation is
performed similarly, but Algorithm~\ref{alg:cond-as} is called with the
value of $1-c$ instead of $c$. 

Lastly, if the pointer is updated in both clauses of the if-then-else
statement, the pointer content prior to that statement needs to be
disregarded. The pointer values used in the two assignments are then merged
as in Algorithm~\ref{alg:cond-as} using the result $c$ of private condition
evaluation. To better illustrate this, consider the following code segment:

{\small \begin{verbatim}
1. p = p1;
2. if (priv-cond) then p = p2;
3. else p = p3;
\end{verbatim}}
\noindent After we assign \texttt{p1} to \texttt{p} on the first line,
\texttt{p}'s content is be overwritten with the content of either
\texttt{p2} or \texttt{p3} depending on the result $c$ of evaluating
\texttt{priv-cond}. We can see that before entering the if-clause, the
current content of \texttt{p} (i.e., that copied from \texttt{p1}) can be
safely disregarded without affecting its correctness. In other words, to
update \texttt{p} inside the conditional statement, we call ${\sf
CondAssign}(\langle \alpha_2, L_2, T_2 \rangle, \langle \alpha_3, L_3, T_3
\rangle, c)$ in Algorithm~\ref{alg:cond-as}, where $\langle \alpha_2, L_2,
T_2 \rangle$ and $\langle \alpha_3, L_3, T_3 \rangle$ are contents of
pointers \texttt{p2} and \texttt{p3}, respectively.

These constructions compose in presence of nested conditional statements with
private conditions. For instance, after executing the code:

{\small \begin{verbatim}
1. if (priv-cond1) then p = p1;
2. else
3.     p = p2;
4.     if (priv-cond2) then p = p3;
5.     else p = p4; 
\end{verbatim}}
\noindent \texttt{p} will contain the combined content of pointers
\texttt{p1}, \texttt{p3}, and \texttt{p4}. That is,
Algorithm~\ref{alg:cond-as} is first called with the content of pointers
\texttt{p3} and \texttt{p4} and the result $c_2$ of evaluating
\texttt{priv-cond2}, after which Algorithm~\ref{alg:cond-as} is called on
the result of its previous execution, the content of \texttt{p1}, and the
result $c_1$ of evaluating \texttt{priv-cond1}.

As evident from the description above, all modifications to variables of
all types (including pointers as well as data) inside conditional statements
with private conditions require special handling inside the compiler. For
each such conditional statement, PICCO examines the list of variables
modified inside the body of the statement and updates them differently from
when the modification is not surrounded by a private condition. Thus, in the
case of pointers we specify how pointers need to be updated inside such
statements using Algorithm~\ref{alg:cond-as} and compiler will process all
variables inside the body of conditional statements with private conditions.

Note that each pointer starts with a single location (i.e., $\alpha$
is set to 1) at the time of its declaration, and the list and the
number of locations $\alpha$ are updated during pointer assignments as
described above. This information is maintained only during program
execution and thus the locations that a pointer might store or their
number is not necessarily known at compile time.

\medskip \noindent \textbf{Pointer dereferencing.} When pointer \texttt{p}
with a private location is being dereferenced, its dereferenced value is
privately computed from $\alpha$, $L = \{\ell_1, {\ldots}, \ell_\alpha\}$,
and $T = \{[t_1], {\ldots}, [t_\alpha]\}$ stored at \texttt{p}. Let $[a_i]$
denote the value stored at location $\ell_i \in L$. Then we compute the
dereferenced value as $[v] = \sum_{i=1}^\alpha [a_i] \cdot [t_i]$. Note that
with linear secret sharing this operation can be implemented as an inner
product that costs only a single interactive operation resulting in a
profound impact on performance.

When the dereferenced value is being updated, all locations in $L$ need to
be touched, but the content of only one of them is being changed. If we, as
before, use $[a_i]$ to denote the value stored at $\ell_i \in L$ and let
$[a_{new}]$ denote the value with which the dereferenced value is being
updated, then we update the content of each location $\ell_i$ as $[a_i] =
[t_i] \cdot [a_{new}] + (1-[t_i]) \cdot [a_i]$. That is, the true location
($t_i = 1$) will be set to $a_{new}$, while all others ($t_i = 0$) will be
kept unchanged.

In the current form, the above procedures are applicable only to pointers
with the indirection level equal to 1. That is, if pointer \texttt{p} is
associated with a list of private locations of pointers, the above
computation will result in producing secret shared locations and the
information looses its semantic meaning. Thus, for pointers with indirection
level $>1$ different computation is used. That is, now each $\ell_i \in L$
stores an address of a pointer $p_i$ and let each $p_i$ be associated with
$\alpha_i$, $L_i = \{\ell^{(i)}_1, {\ldots}, \ell^{(i)}_{\alpha_i}\}$, and
$T_i = \{[t^{(i)}_1], {\ldots}, [t^{(i)}_{\alpha_i}]\}$. To retrieve the
dereferenced value of \texttt{p}, we first compute $[t_i] \cdot [t^{(i)}_j]$
for $1 \le i \le \alpha$ and $1 \le j \le {\alpha_i}$ and merge all lists
$L_i$ for $1\le i \le \alpha$. The resulting list is thus set to $L ' = L_1
\cup L_2 \cup {\cdots} \cup L_{\alpha}$ and let $\alpha' = |L'|$. For any
location in $L'$, we compute its corresponding tag as the sum of all $[t_i]
\cdot [t^{(i)}_j]$ values matching that location in the individual lists
$L_i$. (We can simply use the sum because only one tag can be set to 1.) The
result is $\alpha'$, $L'$ and the corresponding tags $T'$.

To illustrate this on an example, let $\alpha = 3$, $T = ([0], [0], [1])$,
and $L$ store the addresses of pointers $p_1$, $p_2$, $p_3$ with $\alpha_1 =
1$, $L_1 = (123)$, $T_1 = (1)$; $\alpha_2 = 2$, $L_2 = (189, 245)$, $T_2 =
([0], [1])$; and $\alpha_3 = 3$, $L_3 = (123, 176, 207)$, $T_3 = ([0], [1],
[0])$. The result of this operation is a pointer with $L' = L_1 \cup L_2
\cup L_3 = (123, 176, 189, 207, 245)$, $\alpha' = |L'| = 5$, and $T' = ([0],
[1], [0], [0], [0])$.

To update the dereferenced value of \texttt{p} through an assignment as in 
\texttt{*p = p'}, each pointer $p_i$ stored at address $\ell_i \in L$ needs
to be updated with \texttt{p'}'s information.  In particular, for each $p_i$
each tag $[t^{(i)}_j]$ (for location $\ell^{(i)}_j$) is updated to
$(1-[t_i]) \cdot [t^{(i)}_j]$. We also compute tag $[t_i] \cdot [t'_j]$ for
each location $\ell'_j$ in \texttt{p'}'s list of locations. We then merge
the location list of each $p_i$ with that of \texttt{p'} to form $p_i$'s new
list. For any new location inserted into $L_i$, its tag is set to the
computed $[t_i] \cdot [t'_j]$ for the appropriate choice of $j$, and any
location that appears on both $p_i$ and \texttt{p'} lists, the value
$[t_i] \cdot [t'_j]$ is added to $p_i$'s updated tag for that location. In
other words, if $t_i$ is true, we take \texttt{p'}'s value and otherwise
keep $p_i$'s value.

If pointer \texttt{p} with a private location is being dereferenced $m > 1$
times, the above dereference algorithms are naturally applied multiple times
with the first $m-1$ instances being the version that produces a pointer and
the last instance producing either a pointer or a private value depending on
\texttt{p}'s indirection level. \texttt{p} can then be treated as the root
of a tree with its child nodes being locations of pointers stored in its
list and the leaves of the tree eventually pointing to private data (of a
non-pointer type). To perform an $m$-level dereferencing operation, we
traverse the top $m+1$ levels of the tree and consolidate the values stored
at those levels (and update the values at the $(m+1)$st level if the
dereferenced value is to be updated). 

\medskip \noindent \textbf{Secrecy of pointers to private data.} As previously
discussed, the value of a pointer to private data is treated as public when
it stores a single location ($\alpha = 1$), and it is private otherwise
($\alpha > 1$). More generally, if pointers to private data are used in
predicates or similar expressions, the result of a predicate evaluation is
public if its outcome can be determined using only public data. For example,
the outcome of an expression that compares two pointers to private data for
equality is public if (i) both pointers store a single location in their
lists $L_1$ and $L_2$ or (ii) at least one of the pointers stores multiple
locations, but $L_1 \cap L_2 = \emptyset$. In other circumstances, the
outcome depends on private tags and is treated as private.

Note that when the result of a predicate evaluation on pointers is private,
it can be naturally computed by privately determining the true location of
each pointer and applying the predicate to them. This computation, however,
can be optimized for certain types of predicates to result in faster
performance. For example, in the case of pointer equality, the general
solution is to compute $[\ell_1] = \sum_{\ell^{(1)}_i \in L_1} \ell^{(1)}_i
[t^{(1)}_i]$ and $[\ell_2] = \sum_{\ell^{(2)}_i \in L_2} \ell^{(2)}_i
[t^{(2)}_i]$ and then compare $[\ell_1]$ and $[\ell_2]$ for equality, while
an optimized implementation computes $\sum_{\ell^{(1)}_i \in L_1 \cap L_2}
[t^{(1)}_i] [t^{(2)}_j]$, where $\ell^{(1)}_i = \ell^{(2)}_j$ for each
$\ell^{(1)}_i \in L_1 \cap L_2$. The latter can be implemented as the inner
product that costs a single interactive operation and is significantly lower
than the cost of comparing two private integers.

Because it is not always possible to determine at compile time whether a
predicate evaluated on one or more pointers (to private data) will have a
public or private status (which, for example, may depend on program's public
input), some checking will need to be deferred to run time. In particular,
if pointers to private data are used in a predicate to form a conditional
statement and the result of its evaluation is private, the usual constraints
for the body of conditional statements with private conditions apply. We
address this by evaluating the body of such conditional statements for
public side effects at compile time (as in the original PICCO design). If
the body contains public side effects, the transformed program will include
checks for the status of the conditional statement at run time. If the
result of evaluating the conditional statement is determined to be private
at run time (and public side effects are present in the body of the
statement), execution will be aborted with an error due to a possible
information leak.

Note that the fact whether the execution is aborted or not never leaks
private information, as this decision is solely based on public
data. That is, an abort takes place when (i) the result of evaluating
a predicate on two or more pointers is treated as private and (ii) the
body of the conditional statement that uses the predicate contains
public side effects. Whether the result of evaluating the predicate is
public or not is public knowledge because it is determined by the
public locations stored in the pointers and the predicate
itself. Similarly, whether the body of the conditional statement has
public side effects or not is based on the public code that forms the
body of the conditional statement.

\medskip \noindent \textbf{Optimizations.} Because the computation involved
in computing a pointer's dereferenced value is interactive (and thus is
relatively expensive) when the pointer stores multiple locations, we
considered caching and reusing its result. That is, when a pointer's
dereferenced value is computed, we can store it in the pointer structure and
reuse the value on consecutive dereference operations when there are no
changes to the values stored at all pointer's locations $L$ between such
operations. Once any value stored at one of the pointer's locations is
modified, the cached dereferenced value needs to be marked as out of date. A
similar caching technique can also be applied to the computation on private
tags that takes place during pointer update and dereference operations. Note
that private tags can be viewed as aggregation of conjunctions of 1-bit
private variables that denote the evaluation results of private conditions
(or their negations) in user programs. Because, once computed, the variables
will have fixed values, the conjunction results of those variables can be
cached in our framework in a lookup table and allow for their retrieval when
the same conjunctions need to be repeatedly computed. This optimization will
result in considerable savings when multiple private pointers are updated or
dereferenced within the body of a conditional statement with a private
condition. The savings due to either of the above caching optimizations are,
however, application-dependent and require additional program analysis at
program parse and translation time. Thus, these optimizations are not
presently a part of our implementation.

Another possible optimization can lower a program's memory footprint by
reusing data structures created for pointers to private data. In particular,
when a pointer is assigned another pointer's value (as in \texttt{p = p1}),
we could have both pointers pointing to the same data structure instead of
creating its copy. When, however, one of the pointers with the shared data
structure is being modified, it should be unlinked from the shared data
structure and its data structure modified accordingly. Implementing this
would require that each pointer data structure is associated with a list of
pointer variables which are using it. Furthermore, a data structure can be
reused only when all information stored in it is identical for multiple
pointers (i.e., not only the locations in $L$, but also the private tags in
$T$). Because different pointers often have distinct roles in user programs,
the expected savings are not very large. This optimization is presently not
a part of our implementation.

\subsection{Pointers to Struct}
\label{sec:struct}

We next discuss design and implementation of pointers to structs, including
their representation and the associated algorithms.  Pointers to complex
data types declared using struct constructs are common for building data
structures such as linked lists, stacks, and trees, and thus pointers to
structs deserve special attention. 

As before, if a complex data type contains no private fields, no
transformations are needed. However, when dealing with pointers to struct
with private fields, we need to address the following questions:
\begin{enumerate}
  \item A struct groups together a number of different variables that can be
    either private or public, but the complex data type itself declared
    using struct is not associated with any particular type of secrecy. When
    declaring a pointer to a complex data type, we thus need to
    determine if a pointer to it can be treated as a pointer to private
    data or if it has to be treated as a conventional pointer to a public
    variable.
  \item When designing representation of a private pointer that points to
    struct, we need to take into account the fact that fields of a complex
    data type can be accessed and modified independently of each other or
    the struct itself. Thus, it remains as a question whether we should
    maintain a separate list of addresses for each struct field or maintain
    only a single list of addresses for all possible struct variables
    associated with the pointer.
  \item The last question is whether we can reuse the previously described
    algorithms for working with private pointers for updating or
    dereferencing pointers to structs on the individual fields of a struct
    or if modifications are needed.
\end{enumerate}    
In what follows, we thus focus on answering these questions.
    
\medskip \noindent \textbf{Secrecy of pointers to struct.} Secrecy of a
pointer to struct is implicitly determined by the protection modes of the
struct's fields. We determined that a pointer to a complex data type can be
treated as a pointer to private data only if all fields in its declaration
are private. It means that if at least a single field of a struct is public,
pointers to this data type can be of public type only. This treatment is
necessary to eliminate information leakage when pointers to structs are
modified inside conditional statements with private conditions. Consider,
for example, a data type containing one private and one public field. If we
treat a pointer to this data type as a pointer to private data, it can be
modified inside an if-statement with a private condition and have multiple
locations associated with the pointer. However, by dereferencing and
observing the value of the public field, one can determine the true location
of the pointer and thus learn unauthorized information about the result of
the private condition.

Because a complex data type may contain other struct variables as its
fields, the variables in the data type will need to be checked recursively
to determine whether at least one public field is present (with provisions
to skip cycles in the declarations). If none are found, pointers to this
data type are treated as pointers to private data.

\medskip \noindent \textbf{Pointer to struct representation.} To implement
private pointers to structs, we needed to determine whether a single list of
locations is sufficient for all fields of the complex data type (recall that
all fields are private) or separate lists must be maintained. In working to
answer this question, we determine that there is no need to maintain
multiple lists of locations, because the list of locations associated with
each field in the struct must be the same (adjusted for the offset of the
field within the struct). That is, values of a struct's fields can be
modified individually (e.g., as in \texttt{p->x = y}), but the only way to
access or modify the location of a field is through the location of the
entire struct. In other words, the list of addresses associated with a
pointer to struct \texttt{p} (and thus the addresses corresponding to all of
its fields) can be modified only by directly updating \texttt{p}, as
operations of the type \texttt{p->x = y} do not affect the list of addresses
associated with the field \texttt{x}. Storing a single list has the added
benefit that we can employ the same representation of pointers to private
data as for simple data types. This treatment also implies that a pointer to
a struct object will have indirection level 1 even if all fields of the
struct are pointers themselves.

\medskip \noindent \textbf{Operations on private pointers to struct.} We
represent pointers to a struct record in the same way as other pointers.
This means that operations for using pointers and updating their values
remain unchanged. To dereference a specific field of a pointer as in
\texttt{p->x} and retrieve the value of the variable \texttt{x}, also only
minor changes to the previously described algorithms are needed. In
particular, all we need is to determine the offset $f$ of the variable's
address within the record and perform the dereferencing procedure in the
same way as for pointer \texttt{p} itself, but instead of using locations
$\ell_i$ from $L$, we use locations $\ell_i + f$. The same modification
applies to the case when the dereferenced value is modified through
assignment. 

If we would like to dereference \texttt{p} and retrieve the entire record as
in \texttt{rec = *p}, we need to iterate through each field of the struct
and retrieve the dereference value of each field as described above for
\texttt{p->x}. Similarly, to update a dereferenced pointer \texttt{p} as in 
\texttt{*p = rec}, we need to perform the equivalent of \texttt{p->x =
rec.x} for each field \texttt{x} of the struct.

\section{Pointer Uses in Programming} \label{sec:pointer-uses} 

In this section we discuss many common uses of pointers in programming and
how they are translated to our environment of computing with private data.
The topics we cover are passing arguments by reference, dynamic allocation
of memory, array manipulation, and pointer casting. Data structures also
constitute a common use of pointers, but we discuss them separately in
Section~\ref{sec:data-structures}.

\subsection{Passing Arguments by Reference}

Function calls contribute to the basic software engineering principles of
modular program design, but could be expensive in terms of stack memory
usage for the passed arguments. This has led to differentiating between
function calls where the arguments are passed by value and by reference. In
the latter case, the function typically takes a pointer to the argument and
all updates to the dereferenced pointer will be visible after completing the
function call (thus, arguments passed by reference can be used for either
input or output).

Passing private variables to functions by reference inherits the same
benefits as for conventional (public) variables in the programming language.
The good news is that no special provisions are needed for passing private
variables by reference, resulting in efficient implementations. Furthermore,
because often to pass an argument by reference, its address is supplied to a
function call (as opposed to supplying an existing pointer), the resulting
pointer will have a single known location. This allows us to enjoy the
benefits of avoiding using extra resources without the slowdown of working
with pointers with private locations. 

\subsection{Dynamic Memory Allocation}\label{sec:dmem}

Pointers are often used in programming to dynamically allocate memory on the
free store and deallocate it when it is no longer in use. Here we focus on
C-style \texttt{malloc()} and \texttt{free()} used with pointers to public
variables and show what modifications are needed to support dynamic memory
allocation with pointers to private variables.

\texttt{malloc()} in C allocates the requested number of bytes on the heap
which are passed as an argument to the function \texttt{malloc()}. The
result of this function is the address of the allocated variable or the
first array element in case of dynamic array allocation, which is stored in
a pointer. To support dynamic memory allocation for private variables, we
start with the following code in C: 

{\small \begin{verbatim}
1. int* p = (int*) malloc(sizeof(int));
2. int* p1 = (int*) malloc(10*sizeof(int));
\end{verbatim}}
\noindent Here \texttt{p} points to single variable, while \texttt{p1}
points to a dynamic array of size 10. The assignment operator directly saves
the malloc result into the pointer because they are of compatible types.
However, this is not the case for pointers to private variables because a
private pointer is represented using multiple fields. Consequently, we
cannot assign the malloc result directly to a private pointer and use a
modified interface for pointers to private variables. In particular, we use
a function \texttt{pmalloc}\footnote{Note that the choice of the function is
not crucial and it can be called \texttt{malloc} instead to simplify
programmers' effort for transforming an existing program to an equivalent
program that computes with private data. We, however, prefer to use
\texttt{pmalloc} to make it explicit that the computation refers to private
data.} to implement private malloc, which is invoked as:

{\small \begin{verbatim}
1. private int* p = pmalloc(10, private int);
\end{verbatim}}
\noindent As shown, \texttt{pmalloc} takes two arguments,
which are the requested number of dynamic variables and the data type. The
function returns the data structure used for private pointers in our
implementation with $\alpha = 1$ and the only location in $L$ set to the
address of the first variable in the allocated array (when the first
argument to the function is $>1$). Specifying the private data type is
necessary to properly allocate and initialize the memory. For example, in
PICCO a private integer is represented using one variable of type
\texttt{mpz\_t} from the GMP library \cite{gmp} and a private float is
represented using four \texttt{mpz\_t} variables. Once memory for the
necessary number of variables is allocated, each of them also needs to be
initialized before it can be used in computation.

Calling \texttt{free} with a pointer in C allows to deallocate the memory
(for either a variable or dynamic array) to which the pointer is pointing.
To support similar functionality for private variables, we implement a
function \texttt{pfree} that similarly takes a pointer (to a private
variable or dynamic array) as its only argument. With \texttt{pfree} we
distinguish between two different cases: the pointer provided as an argument
to the function has a single known location (i.e., $\alpha = 1$) or it has a
private location out of a public list ($\alpha > 1$).

Handling the first case is simple and efficient: we can simply call the
\texttt{free} command to deallocate memory associated with the address
stored in the pointer. Pointers to private data with public locations are
very common in programs that use pointers to private data or build data
structures from private data (e.g., linked lists, stacks). Freeing memory
used by pointers to private data in such cases is thus going to be extremely
efficient and does not introduce additional overhead. 

Handling the second case well, however, is very challenging. This is because
deallocating physical memory results in publicly observable outcomes, and
we must be extremely careful not to reveal the true location stored in a
pointer with a private location while at the same time reducing the
program's memory usage. For example, a simple strategy of deallocating
memory associated with all locations on a pointer's list of addresses will
not be acceptable for some programs. To illustrate this, consider a dummy
example with two pointers \texttt{p1} and \texttt{p2}, for each of which we
allocate memory using \texttt{pmalloc}. Then the locations to which the
pointers are pointing are swapped based on the result of a private condition
evaluation. We obtain that both \texttt{p1} and \texttt{p2} now contain two
identical locations in their lists of addresses, but their true addresses
are distinct. Suppose we process the data to which \texttt{p1} points and
want to deallocate the corresponding memory. If we deallocate both addresses
on \texttt{p1}'s list, \texttt{p2} becomes a dangling pointer and the data
to which it was pointing is no longer accessible. Thus, such an
implementation of \texttt{pfree} would be too restrictive to permit its
general use. 

Thus, calling \texttt{pfree(p)} should result in deallocating memory
associated with only one address on \texttt{p}'s list of addresses.
Furthermore, the address being deallocated cannot depend on any private data
(but can be any function of public data). This means that we are not
necessarily deallocating memory associated with the true location of the
pointer and other pointers that store the same location on their lists must
be adjusted to preserve correctness of the computation (which involves
additional resources). For example, we can choose to deallocate the fist
location $\ell_1$ on a pointer's list, but if this was not the pointer's
true location (which we can privately check), the data stored at $\ell_1$
needs to be relocated and other pointers storing $\ell_1$ on their lists
need to be updated accordingly. We next describe in more detail how we can
realize this idea.

First, if the pointer \texttt{p} on which \texttt{pfree} was called contains
the default location corresponding to uninitialized pointers on its list of
addresses $L$ (which is public knowledge), we choose not to perform memory
deallocation. This is to ensure that no memory is being deactivated (which
may be in use by other pointers) if \texttt{p} happens to be uninitialized.
Otherwise, we free the first location $\ell_1$ on \texttt{p}'s
list.\footnote{In the event that data stored at the locations contained in
the pointer have different sizes, the location with the data of the smallest
size should be chosen instead.} (Alternatively, the location used by the
smallest number of pointers can be freed.) Before we can actually free the
memory, we need to privately update the values stored at the remaining
locations in $L$ using the value stored at $\ell_1$ to maintain correctness.
We will need to ensure that (i) if $\ell_1$ happens to be the true location,
the values stored in the remaining locations will remain unchanged and (ii)
if $\ell_1$ is not the true location, the value stored at $\ell_1$ can be
found at \texttt{p}'s true location, while the values stored at all other
locations remain unchanged.  The rationale for doing this as follows: if
$\ell_1$ is indeed \texttt{p}'s true location, no additional work would be
required if this fact was public (i.e., it is the programmer's job to ensure
that freeing \texttt{p} does not affect other variables still in use). If
$\ell_1$, however, was not \texttt{p}'s true location, it may be in use by
other pointers and the value stored at $\ell_1$ needs to be relocated to
\texttt{p}'s true location prior to memory deallocation (and the pointers
that contain $\ell_1$ in their lists need to be updated accordingly).

Let \texttt{p} at the time of calling \texttt{pfree} store $\alpha$, $L =
\{\ell_1, {\ldots}, \ell_{\alpha}\}$, $T = \{[t_1], {\ldots},
[t_{\alpha}]\}$ and $A = \{[a_1], {\ldots}, [a_{\alpha}]\}$ denote values
stored at locations in $L$.\footnote{Although in the current discussion we
assume \texttt{p} is a private pointer that points to a non-pointer data
type, the same idea will apply when \texttt{p} points to a pointer. In
particular, if \texttt{p} points to a pointer, the procedure will include
merging the lists of pointers stored at locations $\ell_1$ and $\ell_i$ and
updating the tags similar to the formula for simple data types. Furthermore,
when \texttt{p} is a pointer to a struct, each field is updated separately
according to its type.} To obliviously update $[a_i]$'s for $2 \le i \le
\alpha$, we compute $$[a_i] = [a_1] \cdot [t_i] + [a_i] \cdot (1-[t_i]).$$ 
This satisfies the above two requirements as follows: if $t_1$ is true
($\ell_1$ is the true location) and thus $t_i$ is false, the result will be
$a_i$ for any $i$; if $t_i$ is true and thus $t_1$ is false, the result will
be $a_1$; if both $t_1$ and $t_i$ are false, the result will be $a_i$.
Surprisingly the formula does not depend on $t_1$.

Second, we need to update private pointers that store the freed location
$\ell_1$ in their lists (and are still in use), but no computation needs to
be performed for pointers that store any of $\ell_2, {\ldots}, \ell_\alpha$
from $L$, but not $\ell_1$ itself. The rationale for doing this as follows:
if $\ell_1$ is indeed \texttt{p}'s true location, no additional work would
be required if this fact was public (i.e., it is programmer's job to ensure
that freeing \texttt{p} does not affect other variables still in use). If
$\ell_1$, however, was not \texttt{p}'s true location, it may be in use by
other pointers and the value stored at $\ell_1$ is moved to \texttt{p}'s
true location prior to memory deallocation. We thus need to replace $\ell_1$
in other pointers' lists with locations that are guaranteed to include the
value originally stored at $\ell_1$ and update the locations' tags
accordingly. Thus, for each pointer \texttt{p'} that stores $\ell_1$ in its
list $L'$, we retrieve $\ell_1$'s position $pos$ in $L'$ and its
corresponding tag $t'_{pos}$. We then replace $\ell_1$ in $L'$ with
$\{\ell_2, {\ldots}, \ell_{\alpha}\}$ and $t'_{pos}$ in $T'$ with
$\{[t'_{pos}]\cdot [t_2], {\ldots}, [t'_{pos}]\cdot [t_{\alpha}]\}$. If any of
$\ell_i$ for $i = 2, {\ldots}, \alpha$ already appears in $L'$, that
location is not included the second time and its tag is set to the sum of
the tag already present in $T'$ for location $\ell_i$ and $[t'_{pos}] \cdot
[t_i]$.

Returning to our example with \texttt{p1} and \texttt{p2}, we have that
prior to calling \texttt{pfree(p1)}, \texttt{p1} stores $\alpha_1 = 2$, $L_1
= (\ell_1, \ell_2)$, $T_1 = (t_1, t_2)$, and \texttt{p2} stores $\alpha_2 =
2$, $L_2 = (\ell_2, \ell_1)$, $T_2 = (t'_1, t'_2)$. Then either $t_1 = t'_1
= 1$ and $t_2 = t'_2 = 0$ or $t_1 = t'_1 = 0$ and $t_2 = t'_2 = 1$. Once
\texttt{pfree(p1)} is called, $\ell_1$ is scheduled for deallocation. If
$t_1 = 1$, no changes take place; otherwise ($t_2 = 1$), the data from
location $\ell_1$ is copied into location $\ell_2$. We obtain that location
$\ell_1$ is being removed from $L'$ (and the corresponding tag $t'_2$ from
$T'$) and location $\ell_2$ is being added to $L'$ with the corresponding
tag $t'_2 \cdot t_2$. Because $\ell_2$ is already present in $L'$, it is
stored once and the tag becomes $t'_1 + t'_2 \cdot t_2$. Thus, we have that
$L'$ now stores a single location and the tag is 1 for any possible set of
original tags. 

Note that the second step of updating pointers that store location $\ell_1$
in their lists is more complex when the pointer \texttt{p} being freed
points to a struct or an array. In those cases, multiple addresses are
processed in this step ($\ell_1$ and other valid addresses that store data
at fixed offsets from $\ell_1$) depending on the type of data to which
\texttt{p} points. In our implementation, we gather all interactive
operations associated with the execution of a call to \texttt{pfree} and
perform them simultaneously in a single round.

If the user program is written correctly (i.e., does not leave dangling
pointers after a call to \texttt{free}), our implementation of
\texttt{pfree} will maintain that for each pointer exactly one location's
tag is set to 1 and all other locations' tags are set to 0. When, however, a
call to deallocate memory corresponding to a pointer results in dangling
pointers, all tags in such pointers can be 0. For that reason, if a call to
\texttt{pfree} causes the number of addresses for some pointer to reduce
to 1, we do not treat the corresponding tag as public. That is, when a
program is not correctly written, opening the value of the tag may reveal
private information, while assuming that the tag is 1 may modify the program's
behavior. Thus, our implementation maintains privacy even in the presence of
programming errors that result in dangling pointers.

We also note that the use of \texttt{pmalloc} or \texttt{pfree} will not be
allowed inside conditional statements with private conditions because these
functions have public side effects.

\subsection{Accessing Array Elements}

The next common use of pointers in programming is manipulating arrays using
pointers. Even for statically allocated arrays, the array name is treated as
a constant pointer that points to the first element of the array. Hence,
arrays and pointers are tightly coupled and pointers are used extensively to
work with arrays.

\medskip \noindent \textbf{Array indexing.} Because arrays are based on
pointers, array indexing also applies to pointers. Thus, we can see
constructions such as \texttt{p = a} and \texttt{p[i]}, where \texttt{p} is
a pointer and \texttt{a} is an array, and need to support them for pointers
to private data. Pointer indexing \texttt{p[i]} with a pointer \texttt{p} to
private data and a public index \texttt{i} is implemented naturally, where
we iterate through all locations in the address list $L$ of \texttt{p},
advance each of them by \texttt{i} multiplied by the size of the data type,
retrieve the data at the determined positions, and combine all of them using
private tags for each location to obtain the result. In other words, the
computation is very similar to that of pointer dereferencing, where instead
of retrieving data at the positions specified in $L$, we advance each
position by \texttt{i} data items. (As C permits the use of negative
indices, when \texttt{i} in \texttt{p[i]} is negative each location in $L$
is decremented by the necessary amount during this operation.)

\medskip \noindent \textbf{Pointers as arrays with known bounds.} In PICCO,
statically allocated arrays of private variables have the array size stored
with them (which is known at the array creation time). Knowing the size of
the arrays allows the compiler to support of a number of important
operations on arrays. Most significantly, this permits the use of private
indexing with arrays, when an element at a private position \texttt{i} is
retrieved from an array \texttt{a} using syntax \texttt{a[i]}. (The size of
the array must be known to support private indexing, regardless of what
technique is used to implement it.) This also permits the use of other
operations such as inner (or dot) products on two arrays, which were
introduced to optimize runtime of compiled programs.

We treat private indexing as an essential part of secure computation with
private data and would like to see it supported for arrays dynamically
allocated on the heap. This means that we would like to offer pointer
indexing \texttt{p[i]} with private \texttt{i} and private pointer
\texttt{p}. The main challenge that we need to overcome is the fact that the
size of the memory pointed by \texttt{p} is not available in C. Furthermore,
a location stored in \texttt{p} may be arbitrary and do not correspond to a
valid memory address (i.e., be unaccessible by the program, correspond to
memory marked as not being in use or any location from the program's stack,
etc.). This means that a pointer can take on many addresses which were not
allocated for variable use and for which the corresponding size cannot be
meaningfully determined (i.e., accessing such addresses would trigger
invalid memory access exceptions in safe programming languages). The size of
properly allocated memory, however, can be determined and utilized to
implement private indexing (and other operations that require array size)
with pointers to private data. In particular, all memory that
\texttt{malloc} allocates on the heap is marked with the size of each
allocated block. Thus, we can use the information that
\texttt{malloc}/\texttt{free} maintain to determine whether a pointer
content falls within a properly allocated memory block, and if it is the
case, access the block's size and use it to implement private indexing.

In more detail, in addition to using private indexing with statically
allocated arrays (as already implemented in PICCO), we permit private
indexing to be used with pointers to private data. The latter is only
successful if the location stored in the pointer\footnote{The current
discussion refers to a single location stored in a pointer, which we view as
the most common use of private indexing. When the pointer contains multiple
locations, the operation is performed on each location separately and the
results are combined in the same way as during pointer dereferencing.} was
allocated via a prior call to \texttt{pmalloc} (and it was not deallocated
during a call to \texttt{pfree}). Because the secure implementation that
PICCO produces makes more calls to \texttt{malloc} than once per call to
\texttt{pmalloc}, the program internally maintains a list of addresses
returned by \texttt{malloc} that correspond to memory requested by the user
program (and an address is taken off the list if it is being freed). Then
when private indexing \texttt{p[i]} is called in the user program and the
pointer stores address $\ell$, we iterate through the list of maintained
addresses. For each such address $l$, we retrieve the corresponding block
size $s$ from the information stored by \texttt{malloc} and check whether $l
\le \ell < l+s$ and the offset of $\ell$ from $l$ is a multiple of the data
type size. If these checks succeed for at least one location on the
allocated address list, $s$ is adjusted for the data type size and is used
as the size of the array to which \texttt{p} points. Note that with this
implementation $\ell$ does not have to correspond to the beginning of the
memory block. Then when $\ell$ is not the address of the beginning of the
array, \texttt{i} can legitimately take negative values.

Under the circumstances when the address $\ell$ does not fall within
any memory block dynamically allocated by the user, private indexing
operation is not performed and the returned result is set to be
secret-shares of 0 (note that, regardless to what value the result is
set, it is not guaranteed to be interpreted as an error). We thus
proceed with the computation despite the error, but send signal
\texttt{SIGBUS}\footnote{Alternatively, custom \texttt{SIGUSR1} or
\texttt{SIGUSR2} can be triggered if the user program is known not to
use it.} and store an error message in a fixed location, so that the
program can catch the signal and act on it. We note that the address
that each call to \texttt{pmalloc} returns is always public
information and the programmer can avoid using invalid addresses.
Ideally, the fact that the private indexing operation cannot be
carried out on the given address is determined before the program is
run, at compile time. Unfortunately, this will not always be possible
and for some incorrectly written user programs the error will not be
triggered until the program is executed (i.e., even programming
languages that perform static analysis of user programs do
array-bounds checking dynamically). The best we can do is to perform
static program analysis at compile time and warn the user about places
where such an error might be possible. 

\medskip \noindent \textbf{Pointer arithmetic.} Pointers can be modified
by setting the address to which they point to the result of an arithmetic
expression evaluation. While in C pointers can be used in arbitrary
expressions similar to the way integer variables are used, only a limited
set of operations on pointer variables is meaningful when they are used to 
store and manipulate addresses within the program. For example, pointer
arithmetic can be relied upon to increment or decrement a pointer value by
an integer amount to move to a different position within an array or between
struct fields. Many other arithmetic operations on pointer variables are not
meaningful, and moving between different variables using pointer arithmetic
is unreliable and error-prone. Thus, in PICCO's default configuration we
chose to disable pointer arithmetic involving pointers to private data in
user programs that the compiler processes. We introduce this as a mechanism
for eliminating a large class of programming errors without constraining
expressiveness of user programs. That is, if we want to change the pointer's
position within an array, instead of using \texttt{p = p-i} or \texttt{p =
p1+4*k+1}, the program will be written as \texttt{p = \&p[-i]} and \texttt{p
= \&p1[4*k+1]}, respectively. We note that disabling pointer arithmetic for
pointers to private data in the default configuration should not be treated
as a limitation of the compiler or our approach, but was a deliberate choice
to reduce programming errors without constraining expressiveness of user
programs.

Nevertheless, turning off pointer arithmetic in PICCO makes it deviate from
standard C. Furthermore, there is a small class of functionalities that
become disabled without pointer arithmetic. One example of such
functionalities is the use of embedded linked lists as, for example,
implemented in the Linux kernel. Embedded linked lists might rely on pointer
arithmetic to move between different fields of a struct. Thus, if the need
for this or similar functionality when working with private data arises in
applications that use PICCO, the compiler can be configured with a
command-line flag to enable support for pointer arithmetic involving
pointers to private data. We implement such functionality as described below.

As mentioned before, in regular C, pointers can store any integer values and
using pointer variables in arithmetic expressions will result in evaluating
the expressions on the integer values stored in such variables. In the
context of pointers to private data, we however distinguish between pointer
objects and (integer) addresses that pointer objects store. Thus, pointer
content is no longer equivalent to integer values, and we support pointer
arithmetic with pointer objects only for the purposes of using pointers to
store and manipulate addresses. Such arithmetic operations can be
categorized into two groups:

\begin{enumerate}
  \item Using pointers to private data in expressions of the type \texttt{p
      + exp} and \texttt{p - exp}, where \texttt{exp} is an expression that
      evaluates to a (public) integer value. This operation produces the
      same output as executing \texttt{\&p[exp]}, i.e., all locations stored
      in \texttt{p} are advanced by the amount of space occupied by
      \texttt{exp} elements of the array.
  \item Using pointers to private data in offset computation as in
    \texttt{p1 - p2}. This operation is straightforward to implement when
    both pointers store a single location. When, however, at least one of
    them has multiple locations, our implementation computes the private
    difference between the true locations of the pointers. This option, in
    our opinion, implements the right semantic value as opposed to other
    variants (such as computing pair-wise differences between all addresses
    stored in the pointers) which bear little meaning.
\end{enumerate}
Note that expressions of the type \texttt{p + exp}, and equivalently
\texttt{\&p[exp]}, where \texttt{exp} evaluated to a private integer, are
not meaningful and not supported. 

\subsection{Pointer Casting}

Variable casting refers to the ability to treat a variable of one type as a
variable of another type. Casting a constant or variable of one type to a
constant or variable of another type typically results in the value being
preserved after the conversion (if possible) even if the two types use
different data representations. This means that conversion is likely to
involve computation. In PICCO, conversion between floating point and integer
values is based on the algorithms given in~\cite{ali13}, while conversion
between integer types of different sizes and floating point types of
different sizes requires minimal to no work (assuming no overflow or
underflow detection is required when casting a value to a shorter
representation).

Pointer casting is handled differently and C is unique in the sense of
allowing pointer-based in-memory casting from one data type to another.
Pointer casting involves no data conversion: the memory is read as is and is
interpreted as a sequence of elements of another type. Thus, pointer casting
is meaningful between a limited number of data types. In order to support
pointer casting in PICCO, we need to resolve the main question: because data
representation of private data types differs from data representation of the
corresponding public data types, we need to determine how to mimic sizes of
public data types when working with blocks of private data without modifying
the data itself. That is, all secret shared values in PICCO are represented
as elements of the same field, which means that, for example, shares of a
16-bit integer and shares of a 64-bit integer have the same bitlength. A
programmer who casts memory storing an array of 64-bit integers to a pointer
to an array of 16-bit integers, however, expects to extract four 16-bit
integers from each 64-bit integer. This means that to meet the programmer's
expectations, private data will need to be processed and assembled in a
different form. We, however, cannot modify the original data because only
the pointer was cast, not the data itself.

Instead of duplicating the memory and performing conversion at the
time of casting, our solution is to do the necessary computation at
the time of pointer dereferencing. This means that we need to record
information about the data type from which casting was performed (to
the data type of the pointer) at the time of casting, but delay
conversion until the pointer is dereferenced. We store casting data
type information with the pointer and use it to extract the relevant
portion of the memory at pointer dereferencing time. Note that in
presence of a sequence of casts, only a single data type needs to be
maintained because the memory layout does not change. 

Because in PICCO simple data types can be defined to
have any bitlength, casting, for example, a pointer of one integer type to a
pointer of another integer type does not guarantee that one data type will
have a bitlength multiple of another. In that case we still calculate what
the relevant portion of the memory is based on the position of the memory
being dereferenced, but the last, partially filled, element might not be
reliably extracted. For example, suppose some memory was filled as a
3-element 30-bit integer array. When it is cast to an array of 20-bit
integers, the fourth elements will be extracted as bits 61--80 of the
original data, while retrieving the fifth element might result in memory
violation because there is not enough data in the original array to fully
form that element.

\subsection{Pointers to Functions}

Similar to pointers to ordinary data types, we need to distinguish between
pointers to functions that will be treated as pointers to private data and
pointers to functions that will be treated as pointers to public data. The
former can be used inside conditional statements with private conditions (as
a result of which they acquire multiple locations and the true location
becomes private) and are restricted to functions with no public side
effects. The latter can contain pointers to functions of any type, but
cannot be modified or dereferenced inside conditional statements with
private conditions. The distinction is made at the time of pointer
declaration using private/public qualifiers with \texttt{void} data type.
That is, by using \texttt{private void *p}, \texttt{p} will be treated
similar to other pointers to private data, while all pointers declared
syntax \texttt{public void *p} will be treated as conventional C pointers.

Private pointers to functions are supported naturally in our implementation.
When a pointer stores a single location, the function is invoked as in
conventional program execution. If, however, a pointer acquires multiple
locations as a result of its modification inside conditional statements with
private conditions, at the time of pointer dereferencing all functions
stored in the pointer will be executed, but only the effects of one of them
will be applied. Conceptually this is the same as executing branching
statements with private conditions: all branches are executed, but only the
effects of one of them are applied depending on the result of private
condition evaluation. That is, when a pointer \texttt{p} storing $\alpha$
locations $L = (\ell_1, \ldots, \ell_\alpha)$ with the corresponding tags $T
= (t_1, \ldots, t_\alpha)$ is being dereferenced, each function $f_i$ stored
at address $\ell_i$ is being invoked. Then each (private) variable $a$ that
$f_i$ modifies is set to $a = a_i \cdot t_i + a_{orig} (1-t_i)$, where
$a_{orig}$ and $a_i$ are its original and newly computed by $f_i$ values. If
$a$ is modified by multiple $f_i$'s with indices ${i_1}$ through ${i_k}$,
its value is updated as $a = \sum_{j=1}^k a_{i_j} \cdot t_{i_j} +
a_{orig}(1- \sum_{j=1}^k t_{i_j})$ (recall that only one $t_i$ can be set to
1).

\section{Analysis}
\label{sec:analysis}

After discussing multiple aspects of private pointer design and its uses in
programming, in this section we summarize the notion of pointer to private
data and operations on it and formally show that program execution that
involves pointers to private data complies with a standard definition of
security used in secure multi-party computation. 

A pointer to private data is defined as a C-style pointer to a private
object storing location information of the object, where a private object
can be one of the following types:
\begin{enumerate}
  \item a primitive private data type (private int, float, etc.);
  \item a composite data type (C-style struct), each element of which is a
    private object;
  \item a function with no public side effects;
  \item a pointer to a private object.
\end{enumerate}
The second and forth categories define a pointer to private data
recursively, which means that a pointer can have any indirection level,
nested struct types, and any combination of primitive, composite, and
pointer data types. 
The previously defined operations on pointers to private objects are:
\begin{enumerate}
  \item Pointer read and update with value $v$, denoted as 
    \textbf{read}(\texttt{p}) and \textbf{update}(\texttt{p}, $v$),
    respectively. 
  \item Dereferenced pointer read and update, denoted as
    \textbf{read}(\texttt{*p}) and \textbf{update}(\texttt{*p}, $v$),
    respectively.
  \item Pointer update inside a conditional statement with a private
    condition. Following prior work, we use multiplexor notation to denote
    this operation as \textbf{mux}(\texttt{p}, $v_1$, $v_2$, $cond$), where
    \texttt{p} is set to $v_1$ if private $cond$ evaluates to 1 and to $v_2$
    otherwise. Pointer read inside a conditional statement with a private
    condition is processed identically to a conventional read
    \textbf{read}(\texttt{p}) and thus is not listed as a separate
    operation. 
  \item Dereferenced pointer update inside a conditional statement with a
    private condition, which we denote as \textbf{mux}(\texttt{*p}, $v_1$,
    $v_2$, $cond$). Similar to the previous case, processing of dereferenced
    pointer reads is not affected by the presence of conditional statements.
  \item Dynamic memory allocation in the form of malloc; for an assignment
    \texttt{p = pmalloc}($n$, $type$) we use notation
    \textbf{alloc}(\texttt{p}, $n$, $type$). 
  \item Dynamic memory deallocation as in \texttt{pfree(p)},
    denoted as \textbf{dealloc}(\texttt{p}).
  \item Array indexing \texttt{p}[$i$] with a public index $i$. This is
    treated as a generalization of pointer dereferencing, and we use notation
    \textbf{read}(\texttt{p}, $i$), \textbf{update}(\texttt{p}, $i$, $v$),
    \textbf{mux}(\texttt{p}, $i$, $v_1$, $v_2$, $cond$) to denote read,
    update, and update inside a conditional statement with a private
    condition, respectively.
  \item Array indexing \texttt{p}[$i$] with a private index $i$ is also
    divided into three operations \textbf{readp}(\texttt{p}, $i$),
    \textbf{updatep}(\texttt{p}, $i$, $v$), and \textbf{muxp}(\texttt{p},
    $i$, $v_1$, $v_2$, $cond$). These operations can only be performed on
    locations stored in a pointer that correspond to arrays with known
    bounds (i.e., allocated using the \texttt{pmalloc} interface or static
    array declaration). 
  \item Evaluation of predicate $f$ on one or more pointers \texttt{p}$_1,
    {\ldots}$, denoted as \textbf{pred}($f$, \texttt{p}$_1$, {\ldots}).
  \item Pointer casting, denoted as \textbf{cast}(\texttt{p},
    $type$), where $type$ is the data type to which \texttt{p} is cast.
\end{enumerate}
In the above, $v$, $v_1$, and $v_2$ correspond to either values associated
with private objects or data structures corresponding to pointers to private
objects, depending on the context. The value of $cond$ is always a private
bit, while variables $n$, $i$, and $type$ are public. Any variable can be
read inside a conditional statement with a private condition, but updates
can be performed only as specified using the \textbf{mux} operations.

This list implicitly defines operations that cannot be performed on pointers
to private objects and will be rejected by the compiler. That is, addresses
of public objects cannot be used to update a pointer to private data (either
via \textbf{update} or \textbf{mux}); a pointer to public data or a mix of
private and private fields defined with a struct construct cannot be
modified inside a conditional statement with a private condition (i.e.,
there is no corresponding \textbf{mux} operation); \texttt{pmalloc} and
\texttt{pfree} cannot be called inside conditional statements with private
conditions (i.e., there is no \textbf{mux} operations for them); similarly,
casting cannot be called inside conditional statements with private
conditions as it modifies publicly stored data.

Recall that our implementation of pointers maintains the invariant that in a
well-formed program there is exactly one true location associated with a
private object. The invariant may be violated only when memory associated
with a pointer is being deallocated when the pointer is still in use (in
such a case, no tag is set to 1 and dereferencing the pointer will return no
data).

In showing security of pointer-related operations that reference private
data, we use a traditional simulation-based definition of security. Because
PICCO is built on top of ($n$, $t$)-threshold secret sharing techniques with
$n \ge 3$ computational parties, we utilize the same setup in our security
analysis. Similarly, because the underlying techniques offer information
theoretic security, we utilize statistical (as opposed to computational)
indistinguishability in the security definition.  

\begin{definition}
Let parties $p_1, {\ldots}, p_n$ engage in a protocol $\Pi$ that evaluates
program $P$ on a mix of public and private data. Let
$\mathrm{VIEW}_\Pi(p_i)$ denote the view of $p_i$ during the execution of
$\Pi$, which is formed by its input, internal random coin tosses $r_i$, and
messages $m_1, {\ldots}, m_k$ passed between the parties during protocol
execution: $\mathrm{VIEW}_{\Pi}(p_i) = (in_i, r_i, m_1, {\ldots}, m_k).$ Let
$I$ denote a subset of the participants of size $t$ and
$\mathrm{VIEW}_\Pi(I)$ denote the combined view of the participants in $I$
during the execution of $\Pi$. Protocol $\Pi$ is said to be $t$-private in
the presence of semi-honest adversaries if for each coalition of size at
most $t < n/2$ there exists a probabilistic polynomial time simulator $S_I$
that given the input of the parties in $I$, $P$, and $P$'s output, produces
a view statistically indistinguishable from $\mathrm{VIEW}_\Pi(I)$ together
with the output of the parties in $I$.
\end{definition}

\begin{theorem} \label{thm:security}
Any program $P$ augmented with pointers to public and private data and no
out-of-boundary access compiled by PICCO is translated into a $t$-private
protocol for any $t < n/2$ when the computation is carried out by $n$
parties.
\end{theorem}

\begin{proof}
Our proof proceeds by evaluating each operation involving a pointer to
a private object as summarized above. After building a simulator for
each pointer-related operation, we apply the composition theorem of
Canetti~\cite{can00} to the result, which would guarantee that any
combination of these operations (and other secure operations in PICCO)
results in security of the overall program $P$ that the computational
parties execute. Building our simulator requires only the use of $t$-private
implementations of addition/subtraction and multiplication operations (which
is met in PICCO by the underlying linear secret sharing scheme with $t <
n/2$).

We describe a simulator for each operation involving a pointer to a private
object in turn. Prior to each operation, each party in the adversarial
coalition $I$ holds a share of each relevant private data item (including
private input data, private fields of pointer data structures, etc.) and at
the time of operation termination each party in $I$ holds a share of the
output and/or updated private items. Many functions also modify data
publicly available to each party (including the simulator).
\begin{itemize}
  \item \textbf{read}(\texttt{p}): This operation simply retrieves the data
    structure stored in \texttt{p} and is local to each computational party.
    The simulator does not interact with the parties in $I$.
    
  \item \textbf{update}(\texttt{p}, $v$): The data structure contained in
    $v$ is simply copied into \texttt{p} by each party. The simulator does
    not interact with the parties in $I$.
    
  \item \textbf{read}(\texttt{*p}): When \texttt{p} stores only a single
    location, the value stored at that location is retrieved and the
    simulator does not interact with the parties in $I$. When
    \texttt{p} stores multiple locations, each party is instructed to
    iterate through the locations extracting values stored in them and
    then combine the values using the private tags associated with
    each location as described in Sections \ref{sec:pimpl}
    and~\ref{sec:struct}. Most of the procedure operates on public
    values (such as locations) and the only private computation
    consists of multiplying tags with private data (or other tags in
    case of pointers to pointers). The simulator thus participates in
    each multiplication operation on behalf of honest users by
    invoking a simulator corresponding to the multiplication operation
    the necessary number of times.
    
    When \texttt{p} is a pointer to an object of a complex data type
    declared using struct and a single field of \texttt{p} is being
    dereferenced (as in \texttt{p->x}), the way the simulator interacts with
    the parties in $I$ is not affected. (Only the locations from which
    values are retrieved are locally modified by a known offset by each
    party.)
    
  \item \textbf{update}(\texttt{*p}, $v$): Similar to reading a referenced
    pointer, when \texttt{p} is associated with only one location, value $v$
    is stored in that location without the parties interacting with the
    simulator. When \texttt{p} stores multiple locations, each party in $I$
    iterates through all possible locations and updates \emph{all locations}
    using secure multiplications as described in Sections \ref{sec:pimpl}
    and~\ref{sec:struct}. Because the locations (and their number) are
    always public, the simulator only needs to engage in a pre-determined
    number of secure multiplications simulating all honest users, which is
    does by invoking a simulator of the secure multiplication protocol.
    
    When \texttt{p} points to a struct and only one field \texttt{p->x} is
    to be updated, the simulator's interaction with the parties in $I$ is
    identical (but the valued are retrieved and the computed values are
    placed not at the locations stored in \texttt{p} but the locations
    adjusted by the offset of \texttt{x} in the struct).
    
  \item \textbf{mux}(\texttt{p}, $v_1$, $v_2$, $cond$): To implement
    conditional assignment, each party (and the simulator) first locally
    merges the lists stored in $v_1$ and $v_2$, after which the (private)
    tags needs to be updated based on the private condition $cond$ as
    described in Algorithm~\ref{alg:cond-as}. The simulator only needs to
    participate in a certain number of secure multiplications determined by
    the public contents of $v_1$ and $v_2$, which it performs as before
    by invoking a multiplication operation simulator.
    
  \item \textbf{mux}(\texttt{*p}, $v_1$, $v_2$, $cond$): Conditional update
    of a dereferenced pointer involves modifying the value stored at each
    location in \texttt{p} with a fixed function of $v_1$, $v_2$ and $cond$
    (see Section~\ref{sec:pimpl}). Thus, the simulator participates in a
    pre-determined number of secure multiplications simulating all honest
    users using a simulator for secure multiplication.
    
  \item \textbf{alloc}(\texttt{p}, $n$, $type$): Allocating memory for a
    number $n$ of private objects of type $type$ does not involve secure
    computation and thus the simulator does not interact with the parties in
    $I$ for the purpose of this operation. 
    
  \item \textbf{dealloc}(\texttt{p}): First, if one of the locations stored
    in \texttt{p} corresponds to the special ``uninitialized'' value, no
    party (including the simulator) performs any operation. Otherwise, the
    first location $\ell_1$ is removed from the data structure that
    \texttt{p} stores and the values at all other locations are updated
    using two secure multiplications (as described in
    Section~\ref{sec:dmem}). The simulator produces communication on behalf
    of all honest users to simulate invocations of the secure multiplication
    protocol as before. The parties consequently locate other pointers that
    store $\ell_1$ and update their locations and tags using a procedure
    that depends only on public data. The simulator is invoked to simulate
    a necessary number of secure multiplications on behalf of honest users
    using the simulator of secure multiplication.
    
  \item \textbf{read}(\texttt{p}, $i$): Similar to a dereferenced pointer
    read, the simulator will need to simulate 0 or more invocations of
    secure multiplications on behalf of honest users. The only difference
    from the simulation of \textbf{read}(\texttt{*p}) is that the pointer
    locations are (locally) adjusted by the space occupied by $i$ items by
    each party (including the simulator) during the computation.
    
  \item \textbf{update}(\texttt{p}, $i$): This procedure is also very
    similar to \textbf{update}(\texttt{*p}), where the difference is only in
    the local (publicly available) computation. 
    
  \item \textbf{mux}(\texttt{p}, $i$, $v_1$, $v_2$, $cond$): Similar to
    \textbf{mux}(\texttt{*p} $v_1$, $v_2$, $cond$), the simulator will need
    to participate in a pre-determined number of secure multiplications
    simulating all honest parties using its simulator after some computation
    on public data.
    
  \item \textbf{readp}(\texttt{p}, $i$): To retrieve an element of an array
    represented by a pointer using a private index, for each location stored
    in \texttt{p} that corresponds to a properly allocated memory block, the
    parties together with the simulator perform a private table lookup. This
    operation is implemented in PICCO by reading each element of the array
    and can be easily simulated once the array size is determined using
    information publicly stored by each computational party. Thus, for each
    eligible address stored in \texttt{p} the simulator simulates private
    table lookup operation on behalf of honest users by invoking a private
    table lookup simulator, after which the results from multiple locations
    (if present) are combined together using private tags, during which the
    simulator engages in secure multiplication simulation.
    
  \item \textbf{updatep}(\texttt{p}, $i$, $v$): This operation proceeds
    similar to \textbf{readp}(\texttt{p}, $i$), but with some operations
    performed in a different order. Each party in $I$ and the simulator
    first locally determine eligible addresses in \texttt{p} using public
    information, after which the parties and the simulator need to modify
    the data being stored for each location using the location's tag through
    secure multiplication and then engages in the interaction for table
    update at a private location. Thus, the simulator needs to engage in a
    pre-determined number of secure multiplication simulations and then in
    simulating interaction corresponding to private table updates. As
    before, the simulator can simply invoke the corresponding simulators for
    secure multiplication and table updates at a private location.
    
  \item \textbf{muxp}(\texttt{p}, $i$, $v_1$, $v_2$, $cond$): Conditional
    update of an array element at a private index is performed similar to
    the regular update. Here, the parties and the simulator can first
    combine data of $v_1$ and $v_2$ using $cond$ into a single $v$, where
    the simulator will need to engage in simulating a fixed number of secure
    multiplications. After this they can follow the steps of
    \textbf{updatep}(\texttt{p}, $i$, $v$) with the simulator producing
    messages as described for the simulation of that operation.
    
  \item \textbf{cast}(\texttt{p}, $type$): This operation only updates
    public information associated with pointer \texttt{p} and the simulator
    does not interact with the parties in $I$.

  \item \textbf{pred}($f$, \texttt{p}$_1$, {\ldots}): The simulator proceeds
    by evaluating the operations in $f$ in order. For any operation that
    uses a single pointer, the simulation proceeds according to one of the
    above cases. For any operation that uses two pointers (most
    significantly, pointer equality or inequality testing), the simulator
    first computes the public/private status of the result. If the status is
    public, the result is determined locally. Otherwise, the simulator
    engages in the computation associated with determining the (private)
    result by invoking the simulators corresponding to the operations used
    in the computation. Once any intermediate result in evaluating $f$ is
    computed to be private, any computation that uses the result is carried
    out on private data by invoking simulators corresponding to the
    operations used in $f$.
\end{itemize}
This covers all possible operations with pointers to private objects and
concludes the proof. 
\end{proof}

We note that our security result above is stated for programs that PICCO
successfully compiles and that do not contain out-of-boundary array accesses.
As previously discussed, the compiler will reject almost all 
types of improperly written programs during static analysis at
compilation time. For example, if a pointer to public data is modified
inside a conditional statement with a private condition, the program
will not compile. This means that the compiler will filter out the
majority of programs with errors, while any well-formed program enjoys
simulatable security.

The main type of errors that the compiler may not be able to detect during
static analysis deal with memory access to invalid locations (e.g., out of
boundary array access, access to a hard-coded invalid location, etc.). When
such accesses are triggered, one or more computational parties might not be
able to continue with the execution and quit on error, but no privacy
violations take place. This is because while incorrect shares of private
data might be read, the data will still remain in the protected form and
cannot be reconstructed without the programmer's original intent. Note that
this discussion applies only to dereferencing a pointer with no offset
(e.g., as in \texttt{*p}) or with a publicly known offset (e.g., as in
\texttt{p}[$i$] with public $i$) because in our implementation accessing
array at a private location never results in out of boundary access (i.e.,
calling this operation on an improperly allocated array will be skipped at
 runtime).

Because the memory layout may differ on different platforms, illegal memory
accesses might result in different behavior of the computational parties
(e.g., execution at one party might result in a segmentation fault faster
than at another party). As our simulator is not guaranteed to run in an
identical environment to those of the honest parties, the simulation might
be distinguishable for programs with illegal memory accesses based on the
point of execution when any given party aborts the computation. Thus, we
exclude such programs from the security statement of
Theorem~\ref{thm:security}, but still guarantee that any program that
compiles in our framework will not result in privacy violations.

To summarize, our solution guarantees that any program that can be compiled
in our framework never reveals any unauthorized information about private
data. Simulation of a poorly written program with illegal memory accesses
may not be indistinguishable from its real execution, but any properly
written program enjoy simulation-based security. 

\section{Pointer-Based Data Structures}
\label{sec:data-structures}

There are several popular data structures typically built using
pointers. In this section we discuss how they would be implemented
using pointers to private data and in what complexities their
performance results.  In particular, we explore linked lists, trees,
stacks, and queues.

\subsection{Linked Lists}
\label{sec:ll}

A linked list consists of a sequentially linked group of nodes. For a singly
linked list, each node is composed of data and a reference in the form of a
pointer to the next node in the sequence, while for more complex variants
such as doubly and circular linked lists the reference field incorporates
additional links. A linked list allows for efficient node insertion and
removal, which makes it an ideal candidate for implementation of stacks and
queues as well as representation of graphs that uses an adjacency list. In
what follows, we discuss implementations of linked lists that store private
data. We start by analyzing various operations in standard linked lists and
then elaborate on the special case when a linked list stores sorted data.
The latter does not represent a typical use of linked lists in programming
(and does not necessarily have attractive features), but is provided as a
relatively simple way to demonstrate what form working with sorted data can
take in a secure computation framework.

\medskip \noindent {\bf Standard linked lists}. Because of ubiquitous use of
linked lists in programming, we analyze different possible uses of linked
lists and the corresponding operations. When a linked list stores public
data, node insertion has cost $O(1)$ as a node is inserted in a fixed place
(e.g., beginning or end of the list). Performing a search requires $O(n)$
time, where $n$ is the number of nodes in the list, because the nodes are
traversed sequentially. Deleting a node from a fixed place (i.e., beginning
or end of the list as done in the case of stacks and queues) involves $O(1)$
time, but when deletion is preceded by a search (and the found node is
deleted), the search together with deletion require $O(n)$ time.

When a linked list stores private data, the reference field holds a pointer
to private data (i.e., a record of the same type) and at the time of node
creation, the pointer stores a single location. Without loss of generality,
let nodes be inserted in the beginning of the list. Node insertion then
places a new node in the beginning of the list manipulating pointers as
before, which still takes $O(1)$ time and is very efficient. Searching a
list involves $n$ private comparisons and all nodes need to be processed as
not to reveal the result of individual comparisons on private data and the
total work is $O(n)$. Similarly, when a node is deleted from the beginning
(or end) of the list, the time complexity of the operation is $O(1)$ and
each node's pointer still stores a single location. It is only when nodes
need to be removed from varying positions in the list and the position
itself needs to be protected, pointers can start acquiring multiple
locations, which causes the time complexity of list traversal and deletion
after a search to go up. However, when the fact whether the searched data
was found in the list or not must remain private, we cannot remove any node,
but instead need to erase the content of a found node (if present) with a
value that indicates ``no data''. In this case, all pointers still contain a
single location and the cost of list search and other operations do not
change, but the list will never reduce in its size. We defer the discussion
of the case when the node is guaranteed to be found in a search and needs to
be removed from a private location until the end of this subsection. 

Throughout this section, we express complexities of the operations as a
function of $n$, where $n$ is the ``visible'' list size. This value will
correspond to the actual list size when delete operations remove an element
from the list so that its size is reduced, while it will correspond to the
number of insertions (i.e., the maximum list size) when delete operations
only mark data as deleted (to hide whether the element was found or not),
but not reduce the list size. The same notation applied to other data
structures discussed in this section when insertions/deletions are based on
the result of private conditions.

\medskip \noindent {\bf Sorted linked lists}. As mentioned before, we
discuss sorted linked lists only as a means of demonstrating how sorted data
might be processed using a general-purpose secure computation compiler and
it should be understood that this is not a typical use of linked lists or
even not the best way of working with sorted data. We use the results of
this discussion in our consecutive description.

Now when a node is being inserted in a linked list, the insertion position
must be determined based on the data stored in the list, which involves
$O(n)$ time with public data (and the complexities of other operations are
the same as before). When we work with private data, the location where the
node is being inserted must remain private (since it depends on private
data) and the execution needs to simulate node insertion at every possible
position. Consider the following two ways of inserting a node and the
performance in which they result:

\begin{enumerate}
\item \emph{Pointer updates:} The first is a traditional implementation of
  node insertion in a linked list, where if the correct insertion point is
  found, we update the pointer of the found node in the list to point to the
  new node and the pointer of the new node to point to the next node in the
  list. Because this conditional statement is based on private data, this
  will result in adding one location to the pointer in the found node and
  one location to the pointer in the node being inserted. After executing
  this operation for every node of the list, the pointer of each node in the
  original list stores 2 locations and the pointer in the newly inserted
  node stores $n$ locations. When this operation is performed repeatedly,
  each node in the list acquires more and more locations (to the maximum of
  the current list size). This means that if the list is built by inserting
  one node at a time, the cost of node insertion and list traversal becomes
  $O(n^2)$. Node deletion after a search also takes $O(n^2)$ time, while
  node deletion from a fixed location is bounded by $O(n)$. When, however,
  only a constant number of nodes are inserted into an existing list (which,
  e.g., can be provided as sorted input into the program), the complexity of
  all operations are unchanged from the public data case.

\item \emph{Data updates:} Another possible implementation of sorted linked
  lists is to always insert a new node at the beginning of the list and keep
  swapping its content with the next node on the list until the correct
  insertion point is found. When this algorithm is implemented obliviously
  on private data using an SMC compiler, the computation processes each node
  on the list starting with the newly inserted node and based on the result
  of private comparison of current and new data either performs the swap or
  keeps the data unchanged. After each node insertion the reference field
  of each node still points to a single node in the list and therefore
  the complexity of all operations are unchanged from their public versions.
\end{enumerate}
Thus, it is clear that we want to avoid acquiring a large number of
locations in each reference field of a pointer-based data structure and
privately moving data (as opposed to privately moving pointers) is preferred
when working with sorted data.

\medskip \noindent \textbf{Node deletion in standard linked lists.}
We can now return to the question of deleting a node from a private location
in a standard (unsorted) linked list when it is known that the searched node
is present and needs to be removed from the list. The above two approaches
of inserting a node in a private position also apply to deleting a node from
a private position. The first, standard, approach of manipulating pointers
will result in acquiring multiple locations at each pointer, which degrades
performance of all operations. Using the second approach of data updates, we
can obliviously place the data to be deleted into the first node on the list
(after scanning the nodes and swapping values based on private data
comparisons) and then simply remove it from the list. This will maintain
optimal complexities of all operations. The above tells us that traditional
implementations of data structures can exhibit performance substantially
worse than alternative implementations in a secure computation framework and
our analysis can be viewed as a step in making informed decisions about
implementation needs. 

\subsection{Trees}

Trees implement hierarchical data structures commonly used to store sorted
data and make searching it easy. A tree node is typically comprised of data
and a list of references to its child nodes. In an $n$-node balanced search
tree, all of searching, node insertion, and node deletion take $O(\log n)$
time. Unfortunately, these complexities greatly change when we write a
program to implement a search tree on private data. In what follows, we
distinguish between trees that are pre-built using the information
available prior to the start of the computation and trees built gradually
using information that becomes available as the computation proceeds: 
\begin{enumerate}
\item \emph{Pre-built trees.} Consider a balanced binary search tree and
  suppose that we want to perform a search on the tree. A traditional
  implementation involves $O(\log n)$ conditional statements to traverse the
  tree from the root to a leaf choosing either the left or right child of
  the current node. When the data is private, such statements use private
  conditions and thus both branches of the computation must be executed. The
  result is that the sequence of $O(\log n)$ nested private conditions
  results in executing all possible $O(n)$ branches of the computation and
  touches all nodes in the tree. This is an exponential increase in the
  complexity compared to working with public data, even if we do not consider
  node insertions and deletions that result in node rotations to balance the
  tree (which are discussed next together with gradually-built trees). 

\item \emph{\it Gradually-built trees.} By analogy with inserting nodes into
  a sorted linked list, we can either manipulate pointers to insert a new
  node at the appropriate place in the tree or insert the node in a fixed
  location and move the data in place. The complexity of the latter option
  is $O(n)$ for insertions, deletions, and search and we take a closer look
  at the former. As we traverse the tree looking for the place to insert the
  new node, similar to searching, all nodes will be touched (as a result of
  nested private conditions). Furthermore, because the execution cannot
  reveal the place into which the new node is inserted, pointers in all
  nodes will acquire new locations. If we add computation associated with
  node rotations when the tree becomes unbalanced, pointers will be
  acquiring new locations even faster (to the maximum of $n-1$ per pointer).
  After repeatedly calling insert to gradually build the tree, eventually
  each node will point to all other nodes resulting in $O(n^2)$ complexity
  for insertions, deletions, and searching. Such complexity is clearly
  avoidable and alternative implementations should be pursued.
\end{enumerate}
Search trees represent the worst possible scenario where implementing an
algorithm on private data using a general-purpose compiler incurs an
exponential increase in its runtime compared to the public data counterpart.
As is evident from our discussion of linked lists and trees, searching an
$n$-element store for a single element cannot be performed in less than
linear time using generic techniques, regardless of whether the data is
stored sorted or not. It means that without custom, internally built
implementations of specific data structures it is conceptually simpler and
more efficient to maintain data in unsorted form, use append for insertion
($O(1)$ time), and shift data to implement deletion.

\subsection{Stacks}

A stack is characterized by the \textit{last-in, first-out} (LIFO) behavior,
which is achieved using push and pop operations. It has several fundamental
applications such as parsing expressions (e.g., parsing programs in
compilers), backtracking, and implementing function calls within an
executable program. To the best of our knowledge, despite its popularity,
this data structure has not been studied in the context of secure
multi-party computation before and our analysis and consecutive
implementation of stack that works with private data demonstrate its appeal
for secure computation.

A pointer-based implementation of a stack is built using a linked list,
where a node is always inserted at the head of the list and is always
removed from the head as well, either of which takes $O(1)$ time. As was
discussed in section~\ref{sec:ll}, implementing these operations on private
data maintains constant time complexities.

When using a stack with private data, we also consider the possibility that
push and pop operations might be performed inside conditional statements
with private conditions, in which case it is not publicly known whether the
operation takes place and what record might be on top of the stack. Then if
we implement a conditional private push operation by manipulating pointers,
the top of the stack will store $m+1$ locations when the last $m$ push
operations were based on private conditions. Implementing a push operation
is then equivalent to executing the code:

{\small \begin{verbatim}
1. node p = new node();
2. if (priv-cond)
3.    p->next = top;
4.    top = p;
\end{verbatim}} 
\noindent Because both \texttt{p} and \texttt{p->next} store
only a single location at the time of conditional push, merging the lists of
\texttt{p->next} and \texttt{top} takes $O(m)$ time. Similarly, merging the
lists of \texttt{top} and \texttt{p} takes $O(m)$ time.

Implementing a pop operation within a private condition involves executing
code:

{\small \begin{verbatim}
1. if (priv-cond)
2.    temp = top;
3.    top = top->next;
4.    // use temp
\end{verbatim}} 
\noindent The complexity of this operation is dominated by the second
assignment. Because \texttt{top} points to $O(m)$ locations, and the
\texttt{next} field of each of its locations can store $O(m)$ locations as
well, the overall complexity of that assignment is $O(m^2)$. This means that
the worst time complexity of a conditional push becomes $O(n)$ for a stack
containing $n$ records and it is $O(n^2)$ for a conditional pop.

If we instead implement push and pop operations that depend on private
conditions by maintaining a single chain of records (with pointers
containing a single location) and data update, push and pop operations
result in $O(1)$ and $O(n)$ work, respectively. That is, we can always
insert a new node (with data or no data depending on the private condition)
into the stack and take $O(n)$ time during pop to privately locate the first
node with data (and erase the data as necessary).

\subsection{Queues}

Queue is another important data structure used to maintain a set of entities
or events in a specified order which are waiting to be served. We can
distinguish between \textit{first-in, first-out} (FIFO), \textit{last-in,
first-out} (LIFO), and priority queues. Implementing a queue involves
maintaining two pointers: the head and the tail. The head points to the
beginning of the queue, i.e., the element that will be removed by a dequeue
operation, and the tail points to the last element added to the queue using
an enqueue operation. 

Similar to the stack, when enqueue and dequeue operations in a FIFO queue
are implemented on public data or private data outside of private
conditional statements, their complexities are $O(1)$. Their complexities
for enqueue and dequeue operations are also $O(n)$ and $O(n^2)$,
respectively, when implemented through private pointer manipulation (the
implementation needs to maintain two pointers for the head and tail of the
queue, but updating the second pointer does not asymptotically increase the
amount of work) and $O(1)$ and $O(n)$, respectively, when private data
update is used. 

In a priority queue, each node additionally stores priority (which we assume
is private) and dequeue removes a node with the highest priority. The
complexity of priority queue operations depends on the underlying data
structure used to implement it. The best known complexities for public data
are $O(\log n)$ for enqueue ($O(1)$ average case) and $O(\log n)$ for
dequeue using a heap.

Suppose for now that all operations are outside conditional statements with
private conditions. If we use a linked list to store queue nodes, the best
performance can be achieved using $O(1)$ for enqueue and $O(n)$ for dequeue
(i.e., always store a newly inserted node in the beginning and remove the
highest priority node from a private location as a result of the search for
the highest priority element) or $O(n)$ for enqueue and $O(1)$ for dequeue
(i.e., store the list sorted and always remove the first node during
dequeue). Then if the operations depend on private conditions, we can
maintain $O(1)$ for enqueue and $O(n)$ for dequeue if the operations depend
on private conditions using a very similar approach to that of regular
queues and stacks. That is, we always insert an element into the beginning
of the queue as a result of a conditional enqueue (if the condition is
false, the element is empty), and during dequeue we scan the queue for the
highest priority element and erase it from the queue if the condition is
true.

If the underlying implementation is a heap, we insert a new node in a fixed
leaf location and use $O(\log n)$ compare-and-exchange operations to
maintain the invariant of a max-heap to implement enqueue. Realizing
dequeue, however, requires $O(n)$ work because it cannot be revealed what
path was traversed from the root to a leaf (since the path choice depends on
private priorities). Similar to other implementations, we can maintain these
complexities even when enqueue and dequeue are performed as a result of
private condition evaluation.

\subsection{Summary}

Before we conclude this section, we would like to summarize performance of
different data structures that can be implemented on private data using
newly introduced pointers to private data or records.
Table~\ref{tab:data-structures} lists the best performance we could achieve
using a pointer-based implementation of the data structures discussed in
this section. Recall that in all data structures with conditional operations
performance depends on the number of insertions as opposed to the actual
data size.
\begin{table} \small \centering 
\begin{tabular}{|l|c|c|c|} \hline
\hfil Data structure & Insert & Delete & Search \\ \hline
Linked list & $O(1)$ & $O(1)$ & $O(n)$ \\ \hline
Linked list (delete at private location) & $O(1)$ & $O(n)$ & $O(n)$ \\ \hline
Search tree & $O(n)$ & $O(n)$ & $O(n)$ \\ \hline
Stack or queue & $O(1)$ & $O(1)$ & --- \\ \hline
Stack (conditional private push \& pop) or & \multirow{2}{*}{$O(1)$} &
\multirow{2}{*}{$O(n)$} & \multirow{2}{*}{---} \\ 
queue (conditional private enqueue \& dequeue) & & & \\ \hline
\multirow{2}{*}{Priority queue} & $O(1)$ & $O(n)$ & --- \\
\cline{2-4}
& $O(n)$ & $O(1)$ & --- \\ \hline
Priority queue (conditional private enqueue \& dequeue) & $O(1)$ & $O(n)$ & --- \\ \hline
\end{tabular}
\caption{Performance of various data structures using pointers to
  private data.} \label{tab:data-structures}
\end{table}
We note that the complexities in Table~\ref{tab:data-structures} can be
directly linked to the amount of memory consumed by those data structures,
with small fixed constants hidden behind the asymptotic notation.

These data structures can also be evaluated using alternative mechanisms.
For example, our analysis suggests that implementing these data structures
using arrays of private data instead of pointers to private data would
result in the same complexities (which is often the case for public data as
well). Also, utilizing ORAM-based implementation can improve asymptotic
complexity of some (but not all) data structures and can lead to faster
runtime in practice at least for large enough data sets. The most pronounced
benefit of using ORAM will be observed for implementing search trees, where
all operations can be performed in polylogarithmic (in $n$) time (e.g.,
using the solution in \cite{Wang14}). On the other hand, using ORAM for
linked lists can only increase the complexity of its operations (even the
complexity of a delete at a private location following a search cannot be
reduced below $O(n)$). Other data structures that can benefit from
ORAM-based implementations are stacks and queues where the operations that
update the data structures are performed inside private conditional
statements. ORAM techniques, however, involve larger constants behind the
big-O notation than simple operations and their initial setup cost is also
significant. We thus leave a thorough comparison of ORAM vs. pointer or
array based implementations of various data structures in this framework as
a direction of future work.

\section{Performance Evaluation}
\label{sec:perf}

In this section, we report on the results of our implementation and
evaluation of a number of representative programs that utilize pointers to
private data. Because such programs have not been previously evaluated in
the context of secure multi-party computation, we cannot draw comparisons
with prior work. In some cases, however, we are able to measure the cost of
using pointers, or the cost of a pointer-based data structure, in a program by
implementing the same or reduced functionality that makes no use of pointers. 
Note that because PICCO can be used for both secure multi-party computation
and outsourcing, the inputs in these programs can come from one or more
input parties/clients.

The programs that we implemented and evaluated as part of this work are: 
\begin{enumerate}
\item The first program constructs a linked list from private data read
  from the input and then traverses the list to count the number of times a
  particular data value appears in the list. This is a traditional
  implementation of a linked list, where each record with private data is
  prepended to the beginning of the list when building it. The program is
  given in Figure~\ref{fig:ul}. 
  
  We next notice that this program is sub-optimal in terms of its run time
  because it does not utilize concurrent execution capabilities provided in
  PICCO. For that reason, we also implement an optimized version of this
  program. The difference is that all private comparisons during the list
  traversal are executed in a single round using PICCO's batch constructs.

\begin{figure}[t!]\small
\begin{verbatim}
struct node {
   private int data;
   struct node *next;
};
public int count = 128;

public int main() {
   public int i;
   private int array[count], output;
   struct node *ptr, *head = 0;

   smcinput(array, 1, count);
   //construct the list 
   for (i = 0; i < count; i++) {
      ptr = pmalloc(1, struct node);
      ptr->data = array[i];
      ptr->next = head;
      head = ptr;
   }
   //traverse the list
   privaate int val = 10;
   ptr = head;
   for (i = 0; i < count; i++) {
      if (ptr->data == val)
         output = output+1;
      ptr = ptr->next;
   }
   smcoutput(output, 1); 
   return 0;
}\end{verbatim}
\caption{Construction and traversal of a linked list.} \label{fig:ul}
\end{figure}

\item To evaluate pointer-based implementations that work with private data
  maintained in a sorted form, and more generally privately manipulating
  pointer locations vs. obliviously moving data, we build a program for a
  sorted linked list. The functionality of this program is similar to that
  of the first program (i.e., create a linked list and then traverse it to
  count the number of occurrences of a given data item in it) and the
  difference is in the way the list is build. We evaluate two variants of
  the program corresponding to pointer update (PU) and data update (DU) as
  described in section~\ref{sec:ll}. The program for the DU variant is given
  in Figure~\ref{fig:sl-vu}, and the program for the PU variant is given in
  Figure~\ref{fig:sl-di}. 
  
\begin{figure}[t!]\small
\begin{verbatim}
struct node {
   int data;
   struct node *next;
};
public int count = 128;

public int main() {
   public int i, j;
   private int array[count], output, tmp;
   struct node *head, *ptr1, *ptr2;
   smcinput(array, 1, count);

   //construct the list
   head = pmalloc(1, struct node);
   head->data = array[0];
    
   for (i = 1; i < count; i++) {
      ptr1 = pmalloc(1, struct node);
      ptr1->data = array[i];
      ptr1->next = head;
      head = ptr1;

      ptr2 = head;
      for (j = 0; j < i; j++) {
         if (ptr2->data > ptr2->next->data) {
            tmp= ptr2->data;
            ptr2->data = ptr2->next->data;
            ptr2->next->data = tmp;
         }
         ptr2 = ptr2->next;
      }
   }
   //traverse the list
   private int val = 10;
   ptr1 = head;
   for (i = 0; i < count; i++) {
      if (ptr1->data == val)
         output = output+1;
      ptr1 = ptr1->next;
   }
   smcoutput(output, 1); 
   return 0; 
}\end{verbatim}
\caption{Construction and traversal of a sorted linked list (using data update).} \label{fig:sl-vu}
\end{figure}

\begin{figure}[h!]\small
\begin{verbatim}
//global declarations are the same as in Fig. 2

public int main() {
   public int i, j;
   private int array[count], output;
   struct node *head, *ptr1, *ptr2;
   smcinput(array, 1, count); 

   //construct the list
   ptr1 = pmalloc(1, struct node);
   ptr2 = pmalloc(1, struct node);
   ptr1->data = array[0];
   ptr2->data = array[1];
   if (array[0] < array[1]) {
      head = ptr1;
      head->next = ptr2;
   } else {
      head = ptr2;
      head->next = ptr1;
   }
   for (i = 2; i < count; i++) {
      ptr1 = pmalloc(1, struct node);
      ptr1->data = array[i];
      ptr1->next = 0;
      ptr2 = head; 
      if (ptr1->data < ptr2->data){
         ptr1->next = ptr2;
         head = ptr1; 
      }      
      for (j = 0; j < i; j++) {
         if ((ptr2->data < array[i]) && 
             (ptr2->next->data > array[i])) {
             ptr1->next = ptr2->next;
             ptr2->next = ptr1;
         }
         ptr2 = ptr2->next;
      }
      if (ptr2->data < ptr1->data)
         ptr2->next = ptr1;
   }
   //traversal code is the same as in Fig. 2

   // remove the head node
   val = head->data;
   ptr1 = head;
   head = head->next;
   pfree(ptr1);
   smcoutput(val, 1);
   return 0; 
}\end{verbatim}
\caption{Construction and traversal of a sorted linked list (using pointer
  update).} \label{fig:sl-di}
\end{figure}

\item The third program implements mergesort that takes an array of
  unsorted integers as its input. The program makes an extensive use of
  pointers to private data to pass data by reference to a function that
  conditionally swaps two data items based on their values (i.e., performs
  the so-called compare-and-exchange operations). Mergesort was chosen not 
  necessarily because it provides the best performance for an oblivious
  sort, the objective instead was to demonstrate how performance of a
  program that utilizes pointers to private data (and exercises modular
  design of a program) compares to a similar program that does not use
  pointers. We thus also evaluate another version of mergesort that performs
  compare-and-exchange operations in place (without calling any function)
  and makes no use of pointers. The pointer-based mergesort is given in
  Figure~\ref{fig:ms-p} and its non-pointer-based implementation is given in
  Figure ~\ref{fig:ms}. 
  
\begin{figure}[t!]\small
\begin{verbatim}
public int K = 128;

public void swap(private int* A, private int* B) {
   private int tmp;
   if (*A > *B) {
       tmp = *A;
       *A = *B;
       *B = tmp;
   }
}

void mergesort(private int *A, public int l, public int r) {
   public int i, j, k, m, size;
   size = r - l + 1;
   
   if (r > l) {
      m = (r + l)/2;
      [ mergesort(A, l, m); ]
      [ mergesort(A, m + 1, r); ]
      
      for (i = size >> 1; i > 0; i = i >> 1) 
         for (j = 0; j < size; j += 2*i) [
            for (k = j; k < j + i; k++) [
             swap(&A[k+l], &A[k+i+l]); 
         ]
      ]
   }
}

public int main() {
   public int median = K/2;
   private int A[K];
   smcinput(A, 1, K);
   mergesort(A, 0, K-1);
   smcoutput(A[median], 1);
   return 0;
}
\end{verbatim}
\caption{Mergesort median program with pointers.} \label{fig:ms-p}
\end{figure}
\begin{figure}[t!]\small
\begin{verbatim}
public int K = 128;
private int A[K];

void mergesort(public int l, public int r) {
   public int i, j, k, m, size;
   size = r - l + 1;
   int tmp[size];
   
   if (r > l) {
      m = (r + l)/2;
      [ mergesort(l, m); ]
      [ mergesort(m + 1, r); ]
      
      for (i = size >> 1; i > 0; i = i >> 1) 
        for (j = 0; j < size; j += 2*i) [
          for (k = j; k < j + i; k++) [
            tmp[k] = A[k+l];
            if (A[k+l] > A[k+i+l]) {
               A[k+l] = A[k+i+l];
               A[k+i+l] = tmp[k];
            }
         ]
       ]
     }
  }
  
public int main() {
   public int median = K/2;
   smcinput(A, 1, K);
   mergesort(0, K-1);
   smcoutput(A[median], 1);
   return 0;
}\end{verbatim}
\caption{Mergesort median program without pointers.} \label{fig:ms}
\end{figure}

\item Our last program implements a shift-reduce parser for a context-free
  grammar (CFG) on private data. This is one of fundamental applications
  that can now be naturally implemented using the compiler by building and
  maintaining a stack, once support for pointers to private data is in place.
  We choose a CFG that corresponds to algebraic expressions consisting of
  additions, multiplications, and parentheses on private integer variables,
  which is specified as follows:
  
  \smallskip
 $\textit{statement} \rightarrow \textit{statement}\ |\ \textit{statement}\ 
 \textbf{+}\ \textit{term}$\\
 $\textit{term} \rightarrow \textit{term}\ |\ \textit{term}\ \textbf{*}\ 
 \textit{factor}$\\
 $\textit{factor} \rightarrow \textbf{var}\ |\
 \textbf{(}\textit{statement}\textbf{)}$

  \smallskip Here, all variables are shown in italics, while terminals are
  set in bold font. The grammar can obviously be generalized to more complex
  expressions and programs that work with private as well as public
  variables of different types. We view this application as enabling one to
  evaluate a custom function on private data without writing and compiling a
  separate program for each function. That is, both the function to be
  evaluated and its input (consisting of private data) are provided as input
  to the parser. We note that it is possible for the function or the grammar
  rules to be private as well, but this would result in an increase in the
  program performance. Our parser uses one lookahead character, and due to
  the complexity of the implementation, the program itself is not included
  in the paper.  
  
  To approximate performance overhead associated with using a pointer-based
  stack, we create a program that performs only arithmetic operations on
  private data which are given to the parser and which the parser executes.
  Note that unlike evaluation of  mergesort, these are not equivalent
  functionalities. That is, one program is much more complex, parses its
  input according to the CF grammar, maintains a stack, etc., while the
  other only performs additions and multiplications.
\end{enumerate}
Note that most of these programs already exercise dynamic memory allocation
(i.e., all linked list programs and the shift-reduce parser). However, to
provide a more complete evaluation of dynamic memory management, we also
include experiments that measure the overhead of dynamic memory
deallocation. Thus, we incorporate calls to \texttt{pfree} to two programs:
(i) we call \texttt{pfree} as part of the shift-reduce parser at the end of
each pop operation and (ii) we evaluate the cost of removing the head node
in a sorted linked list built using pointer update (Figure~\ref{fig:sl-di}).
These were chosen as natural applications of memory deallocation, where
pointers to private objects contain a single and multiple locations,
respectively. In the second case, the head stores locations of all nodes on
the list and the overhead of \texttt{pfree} includes updating the structures
of other pointers on the list upon memory deallocation.

Each program was compiled using PICCO, extended with pointer support as
described in this work, and run in a distributed setting with three
computational parties. All compiled programs utilize the GMP library for
large number arithmetic and OpenSSL to implement secure channels between 
each pair of computational parties. We ran all of our experiments using
three 2.4 GHz 6-core machines running Red Hat Linux and connected through
1Gb/s Ethernet. 
\begin{table*} \centering \small 
\resizebox{\columnwidth}{!}{
\begin{tabular}{|l|c|c|c|c|c|c|c|} \hline

\hfil  \multirow{2}{*}{Program} & Field & \multicolumn{6}{|c|}{Data size}  \\ \cline{3-8}

& size (bits) & $2^5$ & $2^8$ & $2^{11}$ & $2^{14}$ & $2^{17}$ & $2^{20}$ \\ \hline \hline

Linked list & \multirow{3}{*}{81} & 0.0004 $\pm$ $5\%$ &  0.003 $\pm$ $4\%$
&  0.014 $\pm$ $3\%$ & 0.097 $\pm$ $2\%$  & 0.760 $\pm$ $1\%$ &  5.40 $\pm$ $1\%$
 \\  \cline{3-8}
(list building, traversal, & & 0.086 $\pm$ $1\%$ &0.661 $\pm$ $1\%$ & 5.30
$\pm$ $3\%$ & 42.27 $\pm$ $1\%$  & 337.8 $\pm$ $1\%$ & 2,692 $\pm$ $2\%$ \\
\cline{3-8}  

and optimized traversal) & & 0.026 $\pm$ $7\%$ &0.140 $\pm$ $3\%$ & 1.019 $\pm$ $1\%$ & 8.051 $\pm$ $2\%$ & 63.75 $\pm$ $1\%$ &513.8 $\pm$ $1\%$ \\ \hline \hline

Shift-reduce parser &  33 & 0.005 $\pm$ $9\%$ &  0.039 $\pm$ $2\%$  &  0.307 $\pm$ $1\%$  &  2.439 $\pm$ $1\%$ & 19.73 $\pm$ $2\%$ &  157.2 $\pm$ $1\%$
 \\ \hline

Arithmetic operations & 33 & 0.005 $\pm$ $9\%$ & 0.038 $\pm$ $3\%$  &  0.294 $\pm$ $1\%$  &  2.336 $\pm$ $1\%$ & 18.85 $\pm$ $1\%$ & 150.8 $\pm$ $1\%$
 \\ \hline
\end{tabular}
}
\caption{Performance of representative programs with unsorted data structures measured in seconds.} \label{tab:exp1}
\end{table*}

\begin{table*} \centering \small \setlength{\tabcolsep}{1.3ex}
\resizebox{\columnwidth}{!}{
\begin{tabular}{|l|c|c|c|c|c|c|c|} \hline

\hfil  \multirow{2}{*}{Program} & Field &\multicolumn{6}{|c|}{Data size}  \\ \cline{3-8}
& size (bits) & $2^4$ & $2^5$ & $2^6$ & $2^7$ & $2^8$ & $2^9$ \\ \hline \hline

{Sorted linked list (DU)} & \multirow{2}{*}{81} & 0.466 $\pm$ $1\%$ &  1.908
$\pm$ $1\%$  & 7.750 $\pm$ $1\%$  &  31.24 $\pm$ $1\%$  &  125.5 $\pm$ $1\%$  &  565.5 $\pm$ $1\%$ 
 \\  \cline{3-8}
 (list building and traversal) & & 0.036 $\pm$ $1\%$ &0.071 $\pm$ $1\%$ & 0.142 $\pm$ $1\%$ & 0.284 $\pm$ $1\%$  & 0.567 $\pm$ $1\%$  & 1.311 $\pm$ $1\%$ \\ \hline 

{Sorted linked list (PU)}  & \multirow{3}{*}{81} &  1.464 $\pm$ $1\%$ &
9.956 $\pm$ $1\%$  & 85.51 $\pm$ $2\%$ & 918.6 $\pm$ $2\%$ & 9,900 $\pm$ $3\%$ & N/A
 \\  \cline{3-8}
 (list building, traversal, & & 0.051 $\pm$ $1\%$ &0.149 $\pm$ $2\%$ & 0.613 $\pm$ $2\%$ &5.285 $\pm$ $2\%$  &45.93 $\pm$ $2\%$  & N/A\\ \cline{3-8}
 and head node removal) & & 0.005  $\pm$ $1\%$ & 0.015 $\pm$ $1\%$ & 0.044 $\pm$ $2\%$ & 0.174 $\pm$ $2\%$ & 0.720 $\pm$ $3\%$ & N/A \\  \hline \hline

Mergesort without pointers &  81 & 0.053 $\pm$ $5\%$ &  0.121 $\pm$ $5\%$  &  0.271 $\pm$ $5\%$  &  0.625 $\pm$ $4\%$  & 1.453 $\pm$ $4\%$ & 3.124 $\pm$ $4\%$
 \\  \hline

Mergesort with pointers& 81 & 0.053 $\pm$ $4\%$ &  0.122 $\pm$ $5\%$  &  0.272 $\pm$ $5\%$  &  0.638 $\pm$ $6\%$ & 1.503 $\pm$ $5\%$ & 3.201 $\pm$ $5\%$
 \\ \hline \hline
 
\end{tabular}
}
\caption{Performance of representative programs with sorted data structures measured in seconds.} \label{tab:exp2}
\end{table*}

Each experiment was run 10 times, and we report the mean time over all runs
and the corresponding deviation from the mean observed in the experiments.
The results of the experiments for working with unsorted and sorted data are 
given in Tables \ref{tab:exp1} and~\ref{tab:exp2}, respectively. 

As can be seen from the tables, each program was run on data of different 
sizes. For all linked lists programs as well as mergesort, the data size
corresponds to the number of elements in the input set, while for the
shift-reduce parser and arithmetic operations the size corresponds to the
number of arithmetic operations in the formula, which were a mix of 90\%
multiplications and 10\% additions. All linked list experiments contain two
different times, which correspond to the times to build and traverse the
linked list, respectively. The tables also report the size of field elements
in bits used to represent secret shared values. While all programs were
written to work with 32-bit integers, most programs in the table use
statistically secure comparisons, which requires the length of the field
elements to be increased by the statistical security parameter (which we set
to 48). (The size of the field elements needs to be the size
of the data plus one bit to ensure that all data values can be represented.)

The results tell us that working with linked lists (the first program in
Table~\ref{tab:exp1}) in the secure computation framework is very efficient.
That is, building a linked list that consists of thousands of elements takes
a fraction of a second. Traversing a linked list is also rather quick, where
going through a linked list of size $2^{11}$ about 1 second in our optimized
program. 

Performance of the sorted linked lists (the first two programs in
Table~\ref{tab:exp2}) characterizes performance expected from different data
structures where it is necessary to hide the place where a new node or data
item is being inserted. As previously mentioned, there is no good reason to
implement the PU variant of different data structures and it is provided
here for sorted linked lists for illustration purposes only. The DU version
of sorted linked list has the same list traversal time as the regular
(unsorted) linked lists, and the reported time for sorted linked lists can
be further optimized in the same way as it was done for regular linked
lists. When we are building a sorted linked list via DU, each operation
takes $O(n)$ time and thus the time to perform this operation for all $n$
elements of the input is $O(n^2)$. This quadratic performance is also
observed empirically where increasing the size of the data set by a factor
of 2 results in four-time increase in the list building time (all insertion
operations are performed sequentially).  As far as the head node removal
operation in a sorted linked list with PU goes, it consists of two pointer
dereferences (i.e., using data and next fields) and one call to
\texttt{pfree}, where the overhead of \texttt{pfree} was between 76.4\% to
81.3\% of that operation's time. In this particular experiment, each pointer
stores $O(n)$ locations, which contributes to the complexity of both memory
deallocation and pointer dereferencing, but the latter operation can be
performed more efficiently.

If we next look at mergesort (the last two programs in
Table~\ref{tab:exp2}), we see that the variant that uses pointers to private
data and makes a function call to a compare-and-exchange operation for each
comparison and the variant with no pointers and corresponding function calls
differ in their performance by a very small amount. The non-pointer version
that performs less work is faster by 0.4--2.4\%.

Lastly, the performance of our shift-reduce parser (the second program in
Table~\ref{tab:exp1}) is extremely fast and is almost entirely consists of
the time it takes to evaluate the provided formula on private data (the
second program in Table~\ref{tab:exp1}). That is, despite having a more
complex functionality and employing pointer-based stack, the time to perform
arithmetic operations only is almost the same as the time the parser takes.
Also, adding \texttt{pfree} to the program does not effect the runtime
(because the pointer stores a single address) and the times with
memory deallocation are omitted from the table. 

As far as memory consumption goes, the introduction of pointers to private
data only marginally affects the amount of allocated memory for programs
with pointers storing a single memory location (linked list, shift-reduce
parser, and mergesort). The amount of memory needed to store and process
sorted linked lists is quadratic in the data size and matches in its
complexity list traversal. Removing a node from the list and calling
\texttt{pfree} reduces the memory consumed by the data. While in general
calls to \texttt{pfree} can increase memory consumption, in this case all
pointers store the same lists of $O(n)$ locations and removing a node and
merging the lists in \texttt{pfree} decreases the size of each list. In
general, we can say that memory consumption is at most quadratic in the
amount of data and user-declared variables in any program.

We note that all functionalities used for our experiments have alternative
implementations using arrays. For linked lists, mergesort, and a shift-reduce
parser, we expect array-based implementation to exhibit very similar
performance to that based on pointers because all pointers store a single
location. For sorted linked lists, we expect array-based programs to have
performance similar to our data update implementation (with the same
asymptotic complexities). To confirm this finding, we evaluated performance
of array-based sorted linked lists, the source code of which can be found in
Figure~\ref{fig:sl-ar}. 
\begin{figure}[t!]\small
\begin{verbatim}
public int count = 128; 
public int main() {

   public int i, j;
   private int input[count], data[count], a, tmp, output;

   smcinput(input, 1, count);
   
   // build the sorted array
   data[0] = input[0];
   for (i = 1; i < count; i++) {
      // move the data if necessary
      a = input[i];
      for (j = 0; j < i-1; j++) {
         if (a < data[j]) {
            tmp = data[j];
            data[j] = a;
            a = tmp;
         }
      }
      data[i-1] = a;
   }
 
   // traverse the array searching for all instances of the value stored in a
   a = 5;  
   output = 0;
   for (i = 0; i < count; i++) {
      if (data[i] == a) 
         output = output+1;
   }
   smcoutput(output, 1);
   return 0;
}\end{verbatim}
\caption{Construction and traversal of a sorted list (using a static
  array).} \label{fig:sl-ar}
\end{figure}
Building the sorted list took about 20\% less time using arrays for most
data sizes, while list traversal was about 9\% slower using arrays for most
data sizes. Thus, both implementations exhibit comparable performance.
Memory consumption is also similar in most programs, with the exception of
array-based sorted list implementation that uses memory linear in the data
size.

Performance of our pointer-based programs can also be compared to that of
array-based implementations using another system or compiler. Sharemind
\cite{Bogdanov08} is a powerful system that supports a wide range of
programs and, similar to PICCO, builds on (a different type of)
information-theoretic secret sharing, which is hand-optimized to work with
three computational parties. Despite similarities of the setup, Sharemind
programs exhibit significantly different performance characteristics. In
particular, the implementation is optimized for performing a large number of
identical operations in a batch, while the cost of performing only a single
operation is high (e.g., on the order of 100ms for a single integer equality
test \cite{bog12}). As such, Sharemind programs will perform significantly
worse (i.e., orders of magnitude slower) on our programs that perform
sequential execution, such as unoptimized linked list traversal, the
shift-reduce parser, and building a sorted linked list (mergesort is also
slower as reported in \cite{zha13}). In the case of optimized linked list
traversal, on the other hand, Sharemind implementations will still be slower
for small data sets (such as $2^5$), but significantly faster for large data
sets (up to two orders of magnitude faster for $2^{20}$ elements).

All of these experiments demonstrate that pointers have a great potential
for their use in general-purpose programs evaluated over private data. Some
pointer-based data structures can exhibit substantially higher performance
in this framework than their public-data counterparts, and custom,
internally built implementations for such data structures are recommended. 

\section{Conclusions} 
\label{sec:conclusions}

In this work, we introduce the first solution that incorporates support for
pointers to private data into a general-purpose secure multi-party
computation compiler. To maintain efficiency of pointer-based
implementations, we distinguish between pointers with public addresses and
pointers with private addresses and introduce the latter only when
necessary. We provide an extensive evaluation of the impact of our design on
various features of the programming language as well as evaluate performance
of commonly used pointer-based data structures. Our analysis and empirical
experiments indicate that the cost of using pointers to private data is
minimal in many cases. Several pointer-based data structures retain their
best known complexities when they are used to store private data. Complexity
of others (most notably balanced search trees) increases due to the use of
private data flow, and custom, internally built implementations of oblivious
data structures that work with sorted data are recommended. We hope that
this work provides valuable insights into the use of various programming
language features when developing programs for secure computation using a
general-purpose compiler, as well as highlight benefits and limitations of
pointer-based designs for SMC compiler developers.

\section*{Acknowledgments}

We would like to thank Ethan Blanton for discussions at early stages of this
work and anonymous reviewers for the valuable feedback.
This work was supported in part by grants CNS-1319090 and CNS-1223699
from the National Science Foundation and FA9550-13-1-0066 from the Air
Force Office of Scientific Research. Any opinions, findings, and
conclusions or recommendations expressed in this publication are those
of the authors and do not necessarily reflect the views of the funding
agencies.

\bibliographystyle{plain}
\bibliography{PointersBib} 

\end{document}